\documentclass[11pt]{article}
\usepackage{epsfig,amssymb,graphicx,amsmath,amsthm,subcaption}
\usepackage{tikz}
\usetikzlibrary{shapes.geometric, arrows}
\usetikzlibrary{positioning}

\tikzstyle{P} = [circle, minimum width=1.5cm, minimum height=1cm,text centered, draw=black, fill=green!30]
\tikzstyle{S} = [circle, minimum width=1.5cm, minimum height=1cm,text centered, draw=black, fill=blue!30]
\tikzstyle{I} = [circle, minimum width=1.5cm, minimum height=1cm,text centered, draw=black, fill=red!30]
\tikzstyle{W} = [circle, minimum width=1.5cm, minimum height=1cm,text centered, draw=black, fill=gray!30]
\tikzstyle{R} = [circle, minimum width=1.5cm, minimum height=1cm,text centered, draw=black, fill=yellow!30]
\tikzstyle{line} = [thick,-,>=stealth]
\tikzstyle{arrow} = [thick,->,>=triangle 90]
\tikzstyle{openA} = [thick,<-,>=open triangle 90]
\tikzstyle{openB} = [thick,->,>=open triangle 90]
\pgfarrowsdeclarecombine{twotriang}{twotriang}%
{triangle 90}{triangle 90}{triangle 90}{triangle 90}
\tikzstyle{Darrow} = [thick,double,->,>=triangle 90]
\tikzstyle{sarrow} = [thick,->,>=triangle 90,shorten >=2cm]
 
\newtheorem{theorem}{Theorem}[section]
  \newtheorem{lemma}[theorem]{Lemma}

\begin{document}

\title{Time-delayed model of immune response in plants}

\author{G. Neofytou,\hspace{0.5cm}Y.N. Kyrychko,\hspace{0.5cm}K.B. Blyuss\thanks{Corresponding author. Email: k.blyuss@sussex.ac.uk} 
\\\\ Department of Mathematics, University of Sussex, Falmer,\\
Brighton, BN1 9QH, United Kingdom}

\maketitle

\begin{abstract}
In the studies of plant infections, the plant immune response is known to play an essential role.
In this paper we derive and analyse a new mathematical model of plant immune response with particular
account for post-transcriptional gene silencing (PTGS). Besides biologically accurate representation
of the PTGS dynamics, the model explicitly includes two time delays to represent the maturation time of
the growing plant tissue and the non-instantaneous nature of the PTGS. Through analytical and numerical
analysis of stability of the steady states of the model we identify parameter regions associated with recovery
and resistant phenotypes, as well as possible chronic infections. Dynamics of the system in these regimes
is illustrated by numerical simulations of the model.
\end{abstract}

\section{Introduction}

One of the major challenges in supporting a growing human population and satisfying a demand for sustainable food and fuel resources is the understanding of how various diseases affect growth and development of plants \cite{Tilman2011,Strange2005}. From a mathematical perspective, significant efforts have been aimed at qualitative and quantitative analysis of the plant disease dynamics, including the environmental impact and its effects on the global yield of crops. Jeger et al. \cite{Jeger2004} give an overview of some of the quantitative approaches employed in plant virus epidemiology throughout the 20th century. Many mathematical models have focussed on the spread of infection by considering populations of healthy and infected plants, with disease transmission occurring through some intermediary. Since disease propagation in plants is mainly carried out by insect vectors \cite{Purcell2005}, a lot of these models incorporate a vector population either explicitly or through empirically derived  relationships between the carriers and the plant hosts. Other models have investigated the effects of traditional disease controls, such as, roguing and replanting, where any plants carrying a disease are removed and replaced  with healthy new plants \cite{Chan1994,VandenBosch1996,Zhang2012a}. Due to the significant role played by vectors in plant disease transmission, some work has been done on the analysis of their various behaviours, vector aggregation, and the existence of helper viruses that mediate viral transmission \cite{Zhang2000,Zhang2000a}.

Besides vector behaviour, another very important aspect of plant pathology is the dynamics of the plant immune system and regulatory functions. Mathematical models capable of explaining the interactions between plant pathogens and the immune system provide valuable insights into better engineering of genetically modified crops, which are characterized by artificial resistance to specific infections and better adaptation to environmental conditions. Such models also help identify optimal ways of introducing genetically modified crops into the environment in order to minimize any harmful consequences.

Unlike the mammalian immune system, the plant immune system does not possess any form of mobile defence and, therefore, has to rely solely on cellular innate immunity to deal with infections. Moreover, it also exhibits many plant-specific characteristics \cite{JD06}. The plants showing a recovery phenotype to a specific viral infection initially become affected but later experience new growth that is progressively more resistant to the virus until they finally produce new virus-free leaves with complete immunity. An example of such process starts with the plant going into hypersensitive resistance, triggering the self-destruction of infected cells, with necrotic tissue forming at and around the infection site. The cells surrounding these
necrotic lesions are usually found in the antiviral(resistant) state. Although some of them may contain traces of the virus, the virus is unable to replicate \cite{gerg06,frit87}. This can be explained by the {\it post-transcriptional gene silencing} (PTGS),
or RNA interference \cite{Waterhouse1999,Escobar2000,VBF01,Waterhouse2001}, which is a very important gene regulator and a major component of the adaptive immune system. PTGS is characterized by the ability to induce sequence-specific degradation of target messenger RNA (mRNAs) and methylation of target gene sequences. It has been demonstrated that PTGS is mediated by long, perfect or imperfect double-stranded RNAs (dsRNA) produced from either an inverted-repeat transgene or a replicating virus. In the case of viral infection, the core pathway is described as follows: when the viral dsRNA  is injected into the cell, it is targeted by up to four different dicer-like enzymes (DLC),  which chop the viral RNA into short 21-26 nucleotide (nt) long molecules. These molecules, known as short interfering RNAs (siRNA) or microRNAs (miRNA), are used as the building blocks for assembling a special  protein complex called RNA-induced-silencing complex (RISC) that is able to recognise and cleave RNAs containing complementary sequences to the short RNAs  forming their structure. This results in the translation arrest of the viral genome which prohibits viral replication within the cell and prevents the spread of the infection \cite{AlessandraTenorioCosta2013,Hammond2000,Bernstein2001}.
 
Existing mathematical models of PTGS have primarily focused on the intra-cellular aspects of the degradation caused by RNA interference \cite{Bergstrom2003,Raab2004,Groenenboom2005,Cuccato2011}. Most of these models are based on systems of differential equations that describe the dynamics of different RNA populations over time, sometimes also including certain amplification pathways depending on the type of cells being analysed.
Although PTGS has been extensively studied as a gene regulator, its significance as an integral part of the plant immune system so far has not been studied mathematically. In this paper we consider a model of plant disease within a single host with account for PTGS. Instead of explicitly including all the complexity of RNA interference, we propose a slightly simplified approach that is still able to take into account main aspects of PTGS. In fact, we will show that the model can provide an adequate qualitative description of a plant immune response to a viral infection, and support the main types of observable plant disease dynamics, including resistant and recovery phenotypes.

The outline of the paper is as follows. In the next section we discuss the detailed mechanism of PTGS and derive the corresponding mathematical model. In Section 3 we identify all steady states of the model together with conditions for their biological  feasibility. Sections 4 and 5 are devoted to the analysis of stability of these steady state. Section 7 contains results of numerical stability calculations, as well as numerical simulations of the model to illustrate different types of dynamical behaviour that are qualitatively consistent with recovery, resistance and chronic infection phenotypes. The paper concludes with discussion of results and future outlook.

 \section{Derivation of the model}

As a first step in the derivation of a mathematical model for interactions between plant cells and a viral infection, we divide the host population of cells  $N(t)$ into the classes of susceptible cells $S(t)$ consisting of mature cells that are able to induce RNA interference and are susceptible to infection, infected cells $I(t)$ that spread the infection, recovered cells $R(t)$ that are no longer infectious, warned cells $W(t)$ that emerge from susceptible cells upon them receiving the silencing signal, and proliferating cells $P(t)$ that become susceptible to infection after reaching maturity. All possible transitions between these cell populations are illustrated in  Fig.~\ref{sys_dia}. 
 
The effective transmission rate between infected cells and the susceptible cells is given by the parameter $\lambda$, which is taken to be a cumulative parameter accounting for different aspects of the virus life cycle, as well as the actual process of infections. Infected cells are assumed to recover at a rate $\sigma$ as a result of RISC-mediated cleavage or RNA-directed DNA methylation (RdDm) of the viral genome, depending on whether it is an RNA or a DNA virus \cite{Raja2008}. Average mortality rates of non-infected and infected cells are denoted by $\epsilon$ and $z$, respectively, where the infected cells are generally expected to exhibit a reduced lifespan compared to healthy cells, i.e. $z>\epsilon$.
 
A crucial aspect of the PTGS mechanism is that it cannot be maintained indefinitely in all parts of the plant. Laboratory studies have shown that the silenced state cannot be inherited directly, meaning that a parent cell will most likely be unable to produce daughter cells with the anti-viral components needed to deal with the viral infection. It is, therefore, believed that undifferentiated and proliferating cells, e.g meristematic tissue, need to mature or be released from cellular reproduction before they can acquire an antiviral state \cite{Mitsuhara2002}.  Hence, we introduce $P(t)$ as the population of proliferating cells that are responsible for promoting new plant growth. The generation of these new cells depends on the availability of mature cells that are responsible for the collection of nutrients, and the generation rate of new cells will be denoted by $k$.  Recovered cells, although mature, are excluded from contributing to the development of new growth since the loss of function experienced during a viral infection can often cause devastating and irreparable damage to the cell. The proliferating cells have the average maturity time $\tau_1$, after which they are recruited to the susceptible class. The property of non-inheritance is also true for many viral infections, as it is highly unlikely that plant viruses can produce progeny in proliferating cells, in which the silencing state cannot be maintained. One possible explanation for this is the presumed anti-dsRNa activity during cellular mitosis which interferes with the production of dsRNA required for the transmission of PTGS and the replication of a virus \cite{MiassarM.Al-Taleb2011,Taskn2013}. Thus, the cell population $P(t)$ will be assumed to have both immunity to viral infection and the inability to express RNA interference.
 
Evidence suggest that  RISC-mediated cleavage of target transcripts only requires the presence of  21-nt siRNAs, whereas a 25-nt siRNA may also induce RNA methylation and the long-distance transmission of the silencing signal \cite{Molnar2011}. From a molecular point of view, it has been suggested after the initiation of silencing, the primary 2-nt siRNA produced inside the cell can move into surrounding cells regardless of whether they contain any homologous transcripts. In the case where a receiving cell contains homologous transcripts, a second wave of 21-nt siRNA could be synthesized  by using these transcripts as templates. Unlike the first wave, the production of a second wave of siRNA does not require the use of a dicer enzyme but relies on the recruitment of RNA-dependent RNA polymerase (RDR). The importance of this RDR-mediated phase is that it amplifies the silencing signal, and, as a result, these secondary RNAs could be the agents responsible for the systemic movement of the RNA silencing signal \cite{Wassenegger2000,Zhang2012}. In the light of these observations, we consider the class of warned cells $W(t)$ that represent a subgroup of susceptible cells which have successfully acquired immunity to a viral infection by being the recipients of siRNA originating from infected cells. These cells are assumed to express the antiviral components prior to infection, and by doing so, they are capable of degrading the viral genome without any viral interference \cite{Waterhouse2001}.

It is widely understood that  pathogens are capable of eliciting, suppressing or delaying the PTGS response of the plant, and that the induction of PTGS is not  instantaneous \cite{Waterhouse1999,Sonoda2000}. Recent studies have shown that viruses are capable of producing highly specific viral proteins  able  to interfere with the many different stages of the RNA-degradation mechanism \cite{AlessandraTenorioCosta2013,Raja2008,Roth2004,Canizares2008, Llave2000,Burgyan2011}. Taking this into account, in this  model we assume that the propagating signal is initiated by the induction of PTGS in infected cells, but it will, however, be treated independently of whether the infected cells can recover or not. This will allow us to investigate specific cases where a virus can avoid  silencing within the occupied cell but cannot prevent the propagation of the warning signal to other surrounding cells, and vice versa \cite{Zvereva2012}. Hence, the effective warning rate between infected and other target cells will be denoted by $\delta$. We introduce time delay $\tau_2$ to model the average time a cell remains infected before the propagating component of PTGS reaches its target. This is also a cumulative parameter having contributions from viral interference, specific thresholds in dsRNA accumulation necessary for initiation, inherent delay of activation or the transportation delay of involved components. Other infected cells are also assumed to be the recipients of this signal, hence,  $\phi$ will denote the effective rate at which silencing of infected cells can be amplified. Hence, for any time $t$, we assume that $\delta I(t-\tau_2)$ is the signal that has reached susceptible and infected cells, so multiplying this with $S(t)$ gives the number of susceptible cells that become warned by the PTGS signal, whereas multiplying with $I(t)$ and the amplification factor $\phi$ gives the total
number of infected cells that are silenced by the propagating PTGS. This is consistent with the notion of the dsRNA dosage dependence of PTGS: once the virus infects a cell and starts reproducing, it is believed that enough viral dsRNA has to accumulate before PTGS can take place \cite{TF2002}. However, if an infected cell receives additional antiviral components from other neighbouring cells, it is reasonable to assume  that degradation of the viral genome could be initiated either sooner or more efficiently, therefore a stronger immune  response might be possible. 
 
We assume that in the absence of infection the population of susceptible cells should be bounded. To account for this in the model, the population of susceptible cells is taken to decrease at a rate $\epsilon S^2$, where $\epsilon$ is the death rate of non-infected cells, as introduced earlier. Effectively, this corresponds to a logistic growth for susceptible cells, which has been successfully used in other models for the spread of viral infections \cite{Camp61,BK01,PN01,GK05}. Different forms of growth of susceptible cells are discussed in De Leenheer and Smith \cite{LS03} who also provide arguments for only including susceptible cells into the competition term of the logistic growth, in a manner similar to the `Campbell model' \cite{Camp61,BK01}.  From a biological perspective this can effectively account for an additional defensive strategy employed by the plant. Studies suggest that plants afflicted with disease are often able to demonstrate a flexible resource allocation \cite{SJC13,Berg07}. This regulatory function is believed to be a highly complicated process operating through often contradictory channels and is currently not fully understood. However, the core idea is that while pathogens will try to absorb as much nutrients as possible, the plant can dynamically transfer resources from one location to another to either suppress microbial growth or accommodate a defensive response \cite{SJC13,Rojas2014,Ehn1997}. Hence, it seems reasonable to assume that during a viral infection, warned and infected cells will be given priority over resources in order to mount and sustain a proactive and reactive defensive  response respectively, whereas susceptible cells will have to compete with each other.
 
 \begin{figure}[H]
 \centering
\begin{tikzpicture}
\node (start) [P] {$P$};
\node (S) [S, right = 3cm of  start] {$S$};
\node (I) [I, right = 3cm of S] {$I$};
\node (W) [W, above = 1cm of  S] {$W$};
\node (R) [R, right = 3cm of  W] {$R$};
\coordinate[below = 1cm of start] (BP);
\coordinate[right = 1cm of  I] (RI);
\coordinate[below = 1cm of   S] (BS);
\coordinate[below = 1cm of  I] (BI);
\coordinate[above = 1cm of  W] (AW);
\coordinate[above = 1cm of  R] (AR);

\draw [arrow] ([yshift=1ex]start.east) -- node[anchor=south] {$e^{-\epsilon \tau_1}$} ([yshift=1ex]S.west) ;
\draw [openA] ([yshift=-1ex]start.east) -- node[anchor=north] {$k$}([yshift=-1ex]S.west) ;
\draw [arrow] (S.north) -- node[anchor=east] {$\delta$} (W.south);
\draw [arrow] (S.east) -- node[anchor=south] {$\lambda$} (I.west) ;
\draw [openB] (W.west) -- node[anchor=south]{$k$} (start.north);
\draw [arrow] (I.north) -- node[anchor=east]{$\sigma$} (R.south);

\draw [arrow] (start) -- node[anchor=east]{$\epsilon$} (BP);
\draw [-twotriang] (S) -- node[anchor=east]{$\epsilon$} (BS);
\draw [arrow] (I) -- node[anchor=east]{$z$} (BI);
\draw [arrow] (W) -- node[anchor=east]{$\epsilon$} (AW);
\draw [arrow] (R) -- node[anchor=east]{$\epsilon$} (AR);
\draw [arrow] (I.east) -- (RI)|- node[anchor=west]{$\delta\phi$}(R.east);
\end{tikzpicture}
\caption{A diagram of plant immune response within an extended SIR framework. $P$, $S$ and $W$ denote the populations of immature, mature and warned cell whereas $I$ and $R$  stand for infected and recovered cells respectively. Black and white arrowheads represent the direction of recruitment and contribution rates respectively, from one class of cells to another. Note that the  population of susceptible cells $S$ dies at a rate $-\epsilon S^2$, driven by cells competing for available resources, where $\epsilon$ is the natural death rate of plant cells.}\label{sys_dia}
\end{figure}
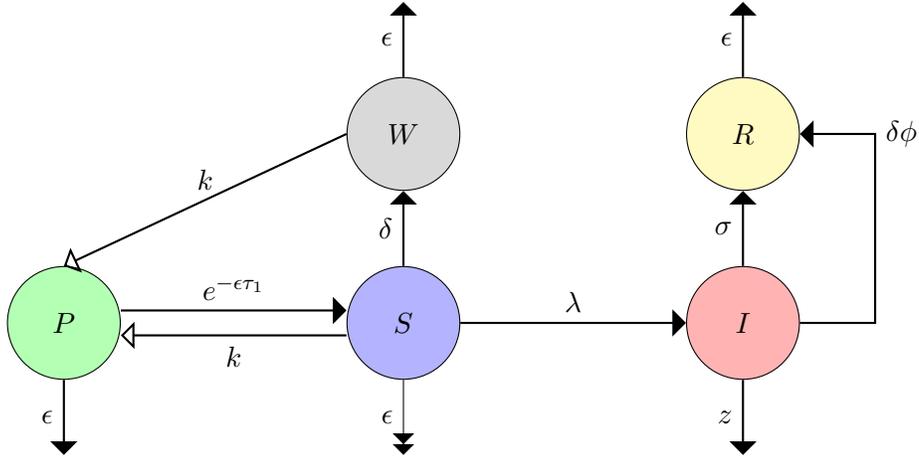

The system describing the dynamics of interactions between plant cells and a viral infection takes the following form
\begin{equation}\label{eq:system_one}
\begin{array}{l}
\displaystyle{\frac{dP}{dt}  = \hspace{0.1cm}  k[S(t)+W(t)] -P_s(t)-\epsilon P(t),}\\\\
\displaystyle{\frac{dS}{dt} = \hspace{0.1cm} P_s(t) - \lambda S(t)I(t)-\delta S(t)I(t-\tau_2) -\epsilon S(t)^2,}\\\\
\displaystyle{\frac{dI}{dt} = \hspace{0.1cm} \lambda S(t)I(t) - (z+\sigma)  I(t)  - \delta\phi I(t) I(t-\tau_2), }\\\\
\displaystyle{\frac{dR}{dt} = \hspace{0.1cm} \sigma I(t) + \delta\phi I(t) I(t-\tau_2) - \epsilon R(t),}\\\\
\displaystyle{\frac{dW}{dt} =\hspace{0.1cm} \delta S(t)I(t-\tau_2) - \epsilon W(t),}
\end{array}
\end{equation}
where $P(t)$, $S(t)$, $I(t)$, $R(t)$ and $W(t)$ denote the populations of proliferating, susceptible, infected, recovered and warned cells, respectively, and
\[
P_s(t)=ke^{-\epsilon \tau_1}[S(t-\tau_1)+ W(t-\tau_1)]
\]
represents the  population of undifferentiated cells that were born at time $t-\tau_1$, have survived for the period of time $\tau_1$ in the class of proliferating cells, and upon maturation move into the class of susceptible cells at time $t$. For biological reasons, system (\ref{eq:system_one}) is augmented with non-negative
initial conditions
\begin{equation}\label{ICs}
\begin{array}{l}
S(s) = S_0(s)> 0, \quad W(s) = W_0(s)\ge 0\hspace{0.3cm}\mbox{for all}\hspace{0.3cm}s \in [-\tau_1,0],\quad P(0)\ge 0,\\\\
I(s) = I_0(s)\ge 0 \hspace{0.3cm}\mbox{for all}\hspace{0.3cm} s \in [-\tau_2,0), \hspace{0.3cm}\mbox{with}\hspace{0.3cm}I(0)>0,\quad R(0)\ge 0.
\end{array}
\end{equation}

\begin{table}[H]
\caption{Table of  parameters}
\begin{tabular}{|  p{.1\linewidth}   p{.58\linewidth}  p{.2\linewidth}| }
\hline
Symbol & Definition & Baseline values\\
\hline 
$\quad\lambda$  & Rate of infection & $1.5$ \\
$\quad k$  & Growth rate & $1$ \\
$\quad\sigma$  & Recovery rate  & $0.5$\\
$\quad\delta$  & Propagation rate of silencing signal & $0.5$ \\
$\quad\phi$     & Amplification factor of recovery & $1$ \\
$\quad\epsilon$ & Natural death rate of cells & $0.3$ \\
$\quad z$        & Death rate of infected cells & $0.6$\\
$\quad\tau_1$   & Maturation time of proliferating tissue & $1$\\
$\quad\tau_2$   & PTGS propagation delay  & $1$\\ 
\hline
\end{tabular}
\label{tab:List of parameters}
\end{table}

The system (\ref{eq:system_one}) can be reduced to the following closed system of equations
\begin{equation} \label{eq:system_two}
\begin{array}{l}
\displaystyle{\frac{dS}{dt} = \hspace{0.1cm} k[S(t-\tau_1)+ W(t-\tau_1)]e^{-\epsilon \tau_1} - S(t)[\lambda I(t)+\delta I(t-\tau_2) +\epsilon S(t)],}\\\\
\displaystyle{\frac{dI}{dt} = \hspace{0.1cm} I(t)[\lambda S(t) - (z+\sigma)    - \delta\phi  I(t-\tau_2)],}\\\\
\displaystyle{\frac{dW}{dt} = \hspace{0.1cm} \delta S(t)I(t-\tau_2) - \epsilon W(t).}
\end{array}
\end{equation}
The remaining two variables $P(t)$ and $R(t)$ are determined by the solutions of this reduced system through
\begin{equation}\label{PRsol}
\begin{array}{l}
\displaystyle{P(t)=P(0)e^{-\epsilon t}+k\int_{t-\tau_1}^{t}[S(x)+W(x)]e^{-\epsilon(t-x)}dx,}\\\\
\displaystyle{R(t)=R(0)e^{-\epsilon t}+\int_0^t[\sigma I(x)+ \delta\phi I(x) I(x-\tau_2)]e^{-\epsilon(t-x)}dx.}
\end{array}
\end{equation}

Since model (\ref{eq:system_one}) and its reduced version (\ref{eq:system_two}) describe the dynamics of cell populations over time, it is essential from a biological perspective for all cell populations to remain non-negative and bounded, as given by the following results.\\

\begin{theorem}
Solutions $P(t),S(t),I(t),W(t),R(t)$ of the system (\ref{eq:system_one}) and $S(t),I(t),W(t)$ of the system (\ref{eq:system_two}) with initial conditions (\ref{ICs}) are non-negative for all $t\geq 0$.
\end{theorem}

\noindent This result can be proven using standard techniques, and it also follows from Theorem 5.2.1 in \cite{Smith95}. The next step is to establish that solutions of system (\ref{eq:system_one}) remain bounded during time evolution.\\

\begin{theorem}
Suppose there exists $T>0$, such that the solution $S(t)$ of the system (\ref{eq:system_one}) satisfies the condition $S(t)\leq M$ for $t\geq T$ with some $M>0$. Then the solutions $P(t),I(t),W(t),R(t)$ of the system (\ref{eq:system_one}) with initial conditions (\ref{ICs}) are bounded
for all $t\geq T$.
\end{theorem}

\begin{proof} Starting with an equation for $I(t)$ in (\ref{eq:system_one}) and using the bound on $S(t)$, we have for $t\geq T$
\[
\frac{dI}{dt}\leq I(t)[\lambda M-\delta\phi I(t-\tau_2)].
\]
Introducing rescaled variables
\[
t=\frac{1}{\lambda M}\tilde{t},\quad I(t)=\frac{\lambda M}{\delta\phi}\tilde{I}(\tilde{t}),\quad \tau_2=\frac{1}{\lambda M}\tilde{\tau}_2,
\]
the above inequality can be rewritten as
\[
\frac{d\tilde{I}}{d\tilde{t}}\leq \tilde{I}(\tilde{t})[1-\tilde{I}(\tilde{t}-\tilde{\tau}_2)].
\]
Proposition 5.13 in \cite{Smith11} together with a comparison theorem implies that the solution of this inequality satisfies
\[
\tilde{I}(\tilde{t})\leq e^{\tilde{\tau}_2},
\]
or, in terms of the original variables,
\[
I(t)\leq \frac{\lambda}{\delta\phi}Me^{\lambda\tau_2 M},
\]
which shows that $I(t)$ is bounded for $t\geq T$. Applying the bounds on $S(t)$ and $I(t)$ in equations for $P$, $W$ and $R$, and using the comparison theorem gives the following results for $t\geq T$
\[
\begin{array}{l}
\displaystyle{P(t)\leq\frac{kM}{\epsilon}\left(1+\frac{\lambda M}{\phi\epsilon} e^{\lambda\tau_2 M}\right)+\left[P(0)-\frac{kM}{\epsilon}\left(1+\frac{\lambda M}{\phi\epsilon} e^{\lambda\tau_2 M}\right)\right]e^{-\epsilon t},}\\\\
\displaystyle{W(t)\leq \frac{\lambda M^2}{\phi\epsilon}e^{\lambda\tau_2 M}+\left[W(0)-\frac{\lambda M^2}{\phi\epsilon}e^{\lambda\tau_2 M}\right]e^{-\epsilon t},}\\\\
\displaystyle{R(t)\leq \frac{\lambda M}{\delta\phi\epsilon}e^{\lambda\tau_2 M}\left(\sigma+\lambda Me^{\lambda\tau_2 M}\right)+\left[R(0)-\frac{\lambda M}{\delta\phi\epsilon}e^{\lambda\tau_2 M}\left(\sigma+\lambda Me^{\lambda\tau_2 M}\right)\right]e^{-\epsilon t}.}
\end{array}
\]
These inequalities prove that the solutions $P(t)$, $W(t)$ and $R(t)$ also remain bounded for $t\geq T$.

\end{proof}

\noindent {\bf Remark.} In all numerical simulations, some of which will be presented in Section 6, the solutions of the system (\ref{eq:system_one}) always satisfied the condition that $S(t)$ remains bounded, which, in the light of {\bf Theorem 2.2}, implies boundedness of all other variables.

Since the variables $P(t)$ and $R(t)$ are fully determined by solutions of the system (\ref{eq:system_two}) through expressions given by (\ref{PRsol}), from now on we will focus on the dynamics of reduced system (\ref{eq:system_two}).

\section{Steady states and  feasibility conditions}

The  system (\ref{eq:system_two}) has up to three possible  steady states. For any parameter values it admits a trivial steady state $E_0=(0,0,0)$ that corresponds to all cell populations going extinct. Linearisation of the system (\ref{eq:system_two}) near the steady state $E_0 = (0,0,0)$ yields a characteristic equation
\begin{eqnarray}\label{CHP_TS}
-(\mu+\epsilon)(\mu+\sigma+z)(k e^{-\epsilon\tau_1}e^{-\mu\tau_1}-\mu)=0.
\end{eqnarray}
Since all parameters are positive, this equation admits two negative roots $\mu_1=-\epsilon$ and $\mu_2=-(\sigma+z)$, and all remaining roots are determined as the solutions of the transcendental equation
\begin{equation}\label{CHP_TS_mu}
\mu = k e^{-\epsilon\tau_1}e^{-\mu\tau_1}.
\end{equation} 
This equation has a real root $\mu>0$ for any values of $k>0$ and $\tau_1\geq 0$, implying that the trivial steady state $E_0$ is unstable for any values of system parameters, and hence, it is impossible for all cell populations to become extinct.

The second steady state of the system (\ref{eq:system_two}) that also exists for any parameter values is a disease-free steady state given by
\begin{equation}
E_1=\left(\epsilon^{-1}K(\tau_1),0,0\right),
\end{equation}
where
\[
K(\tau_1)=k e^{-\epsilon\tau_1}.
\]
It is easy to see that $K(\tau_1)\le k$ for all $\tau_1 \ge 0$.

The third, endemic steady state $E_2 = (S^*,I^*,W^*)$ is characterised by all cell populations being non-zero, and it can be found as
\begin{equation}
\begin{array}{l}
\displaystyle{S^*= S(\tau_1)= \frac{K(\tau_1)}{\epsilon}-\frac{[\delta K(\tau_1)-\epsilon(\lambda+\delta)][\lambda K(\tau_1) -\epsilon (z+\sigma)]}{\epsilon[\epsilon\lambda^2-\delta\lambda(K(\tau_1)-\epsilon)+\delta\phi\epsilon^2]},}\\\\
\displaystyle{I^*= I(\tau_1)=\frac{\epsilon[\lambda K(\tau_1)-\epsilon (z+\sigma)]}{\epsilon\lambda^2-\delta\lambda[K(\tau_1)-\epsilon]+\delta\phi\epsilon^2},}\\\\
\displaystyle{W^* = W(\tau_1)=\frac{\delta\left[\epsilon\delta\phi K(\tau_1)- (z+\sigma)(\delta K(\tau_1)-\epsilon(\lambda+\delta))\right][\lambda K(\tau_1)-\epsilon(z+\sigma)]}{[\epsilon\lambda^2-\delta\lambda( K(\tau_1)-\epsilon)+\delta\phi\epsilon^2]^2}}.
\end{array}
\end{equation}
For the endemic steady state $E_2$ to be biologically feasible, all components $S^*$, $I^*$ and $W^*$ must be positive. It is easy to show that $I^*>0$ implies $S^*,W^*>0$, hence for this steady state to be plausible, it is sufficient to require $I^*>0$.

Let $C = \left\{\frac{\epsilon(z+\sigma)}{\lambda},\frac{\epsilon(\lambda^2+\delta\lambda+\delta\phi\epsilon)}{\delta \lambda}\right\} $ and choose $
C_{\mathrm{min}} = \mathrm{min}(C)$ and $C_{\mathrm{max}} = \mathrm{max}(C)$. Hence, the feasibility condition of the endemic steady state is given by $C_{\mathrm{min}}<K(\tau_1)<C_{\mathrm{max}}$ or equivalently
\begin{equation}\label{conditionEQ}
\frac{\ln(k)-\ln(C_{\mathrm{max}})}{\epsilon}<\tau_1<\frac{\ln(k)-\ln(C_{\mathrm{min}})}{\epsilon}.
\end{equation}
Recalling that $K(\tau_1) = ke^{-\epsilon \tau_1}$, we have $K(\tau_1)\le k$ for all $\tau_1 \ge 0$.
Hence, we have proved the following theorem:\\

\begin{theorem}\label{Theorem:Endemic}
Let the endemic steady state be given by $E_2 = (S^*,I^*,W^*)$. Then the following statements hold:\\
\noindent (i) For $k\le C_{\mathrm{min}}$,  we have that $E_2$ is not feasible.\\\\
\noindent (ii) For  $C_{\mathrm{min}}<k\leq C_{\mathrm{max}}$, the endemic steady state exists if and only if $\displaystyle{\tau_1<[\ln(k)-\ln(C_{\mathrm{min}})]/\epsilon}$.\\\\
\noindent (iii) For $k>C_{\mathrm{max}}$, $E_2$ is feasible if and only if the condition (\ref{conditionEQ}) is satisfied.
\end{theorem}

The conditions of this theorem imply that whilst the trivial and the disease-free steady states exist for any parameter values, the endemic steady state can only exist, provided the growth rate of new plant cells is sufficiently large. This is needed to ensure that a sufficient number of new infections occur before the infection is cleared by the immune response.

\section{Stability analysis of the disease-free steady state}

The characteristic equation of linearisation near the disease-free steady state $E_1$ is given by
\begin{equation}\label{chary2}
\begin{array}{l}
\displaystyle{(\mu+\epsilon)\left[\mu+\sigma+z-\frac{\lambda K(\tau_1)}{\epsilon}\right]\left[-2K(\tau_1)-\mu+K(\tau_1)e^{-\mu\tau_1}\right]=0.}
\end{array}
\end{equation}
One eigenvalue $\mu_1=-\epsilon$ is always negative. The second eigenvalue
\[
\mu_2= \frac{\lambda K(\tau_1)}{\epsilon}-(z+\sigma),
\]
is negative for $\tau_1=0$ if
\[
k<k_{\rm min},\hspace{0.5cm}k_{\mathrm{min}} =\frac{\epsilon(\sigma+z)}{\lambda},
\]
and for $\tau_1>0$, if
\[
\displaystyle{k>k_{\rm min},\hspace{0.5cm} \tau_1>\frac{\ln(k)-\ln(k_{\mathrm{min}})}{\epsilon}}.
\]

The last eigenvalue of the characteristic equation (\ref{chary2}) satisfies the transcendental equation
\begin{equation}\label{mu3}
\mu_3= K(\tau_1)(e^{-\mu_3\tau_1}-2).
\end{equation}

For $\tau_1 = 0$, we have $\mu_3=-k<0$. For $\tau_1 > 0$, it immediately follows that $\mu_3=0$ is not a root of (\ref{mu3}), so we look at the roots of this equation in the form $\mu_3= iw$, $w>0$. Substituting this into (\ref{mu3}) yields
\begin{equation}
K(\tau_1)[\cos(w\tau_1)-2]-iK(\tau_1)\sin(w\tau_1) = 0.
\end{equation}
Since $[\cos(w\tau_1)-2]<0$ for any $\tau_1>0$, this implies that the equation (\ref{chary2}) does not admit purely imaginary roots. Hence,  we have proved the following result:\\

\begin{theorem}\label{Theorem:Disease_free} Let the disease-free steady state be given by $\displaystyle{E_1= \left(\frac{K(\tau_1)}{\epsilon},0,0\right)}$  and denote  $\displaystyle{k_{\mathrm{min}} =\frac{\epsilon(\sigma+z)}{\lambda}}$. Then, the following statements hold:\\
\noindent (a) Given $k<k_{\mathrm{min}}$,  $E_1$ is linearly asymptotically stable for all $\tau_1 \ge 0$.\\\\
\noindent (b) Given $k \ge k_{\mathrm{min}}$ and $\displaystyle{\tau_{min} = \frac{\ln(k)-\ln(k_{\mathrm{min}})}{\epsilon}}$,  $E_1$ is linearly asymptotically stable for $\tau_1>\tau_{min}$, unstable for $\tau<\tau_{min}$, and undergoes a steady-state bifurcation at $\tau_1=\tau_{min}$.
\end{theorem}

This theorem indicates that the disease-free steady is stable, as long as new infections appear slower than they are cleared by recovery or death of the infected cells. Additionally, the theorem suggests that stability of the disease-free steady state depends only on the maturation time of undifferentiated proliferating cells, natural and infection-induced mortality rates, and the rates  at which infected cells spread the infection and recover. This immediately implies that the propagation of the warning signal and the acquired immunity of uninfected cells  is not enough for a complete recovery of the host. Moreover, this suggests that the propagating component of PTGS  acts only as an amplifier of immune response rather than playing an essential role in recovery. Hence, in plants with a strong localized immune response, suppression of the warning signal  would most likely only delay recovery rather than completely inhibit it. Equivalently, a localized  immune response that is too weak will most likely never lead to a complete recovery despite a potentially strong propagating warning signal.

\section{Stability analysis of the endemic steady state}
Linearisation near the endemic steady state $E_2=\left(S^*,I^*,W^*\right)$ yields the following characteristic equation
\begin{equation} \label{mother}
\begin{array}{l}
\mu^3+p_2(\mu,\tau_1,\tau_2)\mu^2+p_1(\mu,\tau_1,\tau_2)\mu+p_0(\mu,\tau_1,\tau_2)=0,
\end{array}
\end{equation}
where
\[
\begin{array}{l}
p_2=\;-K(\tau_1)e^{-\mu\tau_1}+\delta \phi I^*e^{-\mu\tau_2}+\ I^*(\lambda+\delta)+\epsilon (1+2S^*),\\
p_1= e^{-\mu\tau_1}p_{11}+e^{-\mu\tau_2}p_{12}+e^{-\mu(\tau_1+\tau_2)}p_{13}+p_{14},\\
p_0= e^{-\mu\tau_2}p_{01}+e^{-\mu(\tau_1+\tau_2)}p_{02}+p_{03},
\end{array}
\]
and
\[
\begin{array}{l}
p_{11} =\; -K \left( {\it \tau_1} \right)  \left( \epsilon+\delta\,{\it I^*}\right),\\
p_{12} =\; \delta\phi\,\left( \lambda+{\delta} \right) {{\it I^*}}^{2
}+ \left[  \left( \lambda+2\,\phi\,\epsilon \right) \delta\,S^*+\delta\,\phi\,\epsilon \right] {\it I^*},\\
p_{13} =\;-K(\tau_1)\delta \phi I^*,\hspace{0.5cm}
p_{14} =\;\left( \lambda\,\epsilon+\delta\,\epsilon+{\lambda}^{2}S^* \right) {
\it I^*}+2\,{\epsilon}^{2}S^*,\\
p_{01} =\;\delta \phi\left( {\delta}\,\epsilon+\lambda\ \epsilon
\right) {{\it I^*}}^{2}+ \epsilon \left( 2\,{\epsilon}\phi+\lambda\,
\right) \delta\,S^*{\it I^*},\\
p_{02} =\; -K \left( {\it \tau_1} \right) \delta\,{\it I^*}\, \left( {\it I^*}\,
\delta\,\phi+\phi\,\epsilon+\lambda\,S^* \right),\hspace{0.5cm}
p_{03} =\; {\lambda}^{2}\epsilon\,S^*{\it I^*}.
\end{array}
\]

\subsection{Instantaneous maturity}
As a first step in the analysis, we consider the case when the proliferating cells immediately achieve maturity, i.e. $\tau_1 = 0$. In this case, the equation (\ref{mother}) reduces to 
\begin{equation}\label{eq:Tran}
\mu^3+(a_1e^{-\mu \tau_2} +a_2)\mu^2 +(b_1e^{-\mu \tau_2} +b_2) \mu +(c_1e^{-\mu \tau_2} +c_2)=0,
\end{equation}
where
\[
\begin{array}{l}
a_1 =\;\delta\,\phi{\it I^*},\hspace{0.5cm}
a_2 =\; \left( \lambda+\delta \right) {\it I^*}+ \left( 2\,S^*+1 \right) \epsilon-k,\\
b_1 =\; \delta \phi\left( {\delta}+\lambda \right) {{\it I^*}}^{2}+ \left[  \left( \epsilon-k+2\,\epsilon\,S^* \right) \phi+\lambda\,S^*\right] \delta\,{\it I^*},\\
b_2 =\; \left[  \left( \lambda+\delta \right) \epsilon-\delta\,k+{\lambda}^{2}S^*\right] {\it I^*}-k\epsilon+2\,{\epsilon}^{2}S^*,\\
c_1 =\; -\delta I^* \left[ \left[k\delta -\epsilon\left(\delta+\lambda \right) \right] \phi
\,{\it I^*}+ \left( k\lambda -\,{\epsilon}(2\epsilon\phi+\lambda)\right)S^*+k\phi\,\epsilon\right],\\
c_2 =\; {\lambda}^{2}\epsilon\,S^*{\it I^*}.
\end{array}
\]
When $\tau_2= 0$, the characteristic equation (\ref{eq:Tran}) reduces to
\begin{equation*}
\mu^3+(a_1+a_2)\mu^2+(b_1+b_2)\mu+ (c_1+c_2)=0.
\end{equation*}
By the Routh-Hurwitz criterion, the roots of this equation have negative real part if and only if
\begin{equation}
a_1+a_2>0,\; c_1+c_2>0,\; (a_1+a_2)(b_1+b_2)>c_1+c_2.
\end{equation} 

For $\tau_2 > 0$, we follow the methodology of Ruan and Wei \cite{Ruan2001} and look for the roots of the characteristic equation (\ref{eq:Tran}) in the form $\mu=i\omega$, $\omega>0$, which gives
\begin{equation*}
(ib_1\omega+c_1 - a_1\omega^2)(\cos \omega\tau_2-i\sin \omega\tau_2)-i\omega^3-a_2\omega^2+c_2+ib_2\omega = 0.
\end{equation*}
Separating this equation into real and imaginary parts yields
\begin{equation}\label{eq:Two_p}
\begin{array}{l}
b_1\omega \sin \omega\tau_2 - (a_1\omega^2 - c_1)\cos \omega\tau_2 = a_2\omega^2-c_2,\\
b_1\omega \cos \omega\tau_2 + (a_1\omega^2 - c_1)\sin \omega\tau_2 = \omega^3-b_2\omega.
\end{array}
\end{equation}
Squaring and adding these two equations gives the following equation for the Hopf frequency $\omega$:
\begin{equation}\label{eq:after_a}
\omega^6 +({a_2}^2-{a_1}^2 - 2b_2)\omega^4 +(2c_1a_1 -2c_2a_2+{b_2}^2-{b_1}^2)\omega^2+{c_2}^2-{c_1}^2=0.
\end{equation}
Introducing an auxiliary variable $v = \omega^2$, the last equation can be rewritten as
\begin{equation}\label{eq:cubic}
h(v)=v^3+pv^2+qv+r = 0,
\end{equation}
where
\[
\begin{array}{l}
p = ({a_2}^2-{a_1}^2 - 2b_2),\\
q = (2c_1a_1 -2c_2a_2+{b_2}^2-{b_1}^2),\\
r = {c_2}^2-{c_1}^2.
\end{array}
\]

It is straightforward to see that $h(0) = r$ and $\lim_{v\to \infty} h(v) = \infty$, hence given $r<0$, by the intermediate value theorem $h(v)$ has a zero $v_0 \in (0,\infty)$. To investigate what happens when $r$ is positive, we look at the critical points of the function $h(v)$ as given by:
\begin{equation}\label{roots}
v_{1,2} = \frac{-p\pm\sqrt{p^2-3q}}{3}
\end{equation}
One can  see that for $\Delta = p^2-3q<0$,  the quadratic $h'(v)$ has no real roots and so the function $h(v)$ must be monotonic. For  $\lim_{v\to \infty} h(v) = \infty$,  the function $h(v)$ must also be an increasing function, and since $h(0)= r\ge 0 $, we must have that equation (\ref{eq:cubic}) has no positive real roots.

Suppose that $\Delta \ge 0$. Then, for $v_1 = \frac{-p+\sqrt{\Delta}}{3} $, we have $h''(v_{1,2}) = \pm \sqrt{\Delta}$, and, therefore, $v_1$ is a local minimum  whereas $v_2$ is a local maximum  of $h(v)$. Note that $v_2<v_1$, and hence $v_1<0$ implies $v_2 <0$. If $v_1<0$ is the local minimum and $h(0) = r>0$, $h(v)$ is an increasing function on the domain $[v_1,\infty)$, and hence there are no positive real roots of $h(v)=0$. Equivalently, if $v_1>0$, the function $h(v)$ is increasing in the interval $[v_1,\infty)$, hence a positive root can only exist if $h(v_1) \le 0$. We have, therefore, proved the following lemma.
\begin{lemma}\label{Lemma_A}
Let $v_{1,2}$ be given by (\ref{roots}).\\
\noindent (i) If $r<0$, the equation (\ref{eq:cubic}) has at least one positive root.\\\\
\noindent (ii) If $r \ge 0 $ and $\Delta <0$, or $\Delta>0$ and $v_1<0$, or $\Delta>0$, $v_1>0$ and $h(v_1)> 0 $, the equation (\ref{eq:cubic}) has no positive roots.
\noindent (iii) If $r\ge 0 $, $v_1>0$ and $h(v_1)\le 0$, the equation (\ref{eq:cubic}) has at least one positive root.
\end{lemma}
Without loss of generality, let us assume that equation (\ref{eq:cubic}) has  three distinct positive roots denoted by $v_1,v_2$ and $v_3$. This implies that the equation (\ref{eq:after_a}) also has at least three positive roots
\[
w_1 = \sqrt{v_1},\;w_2 = \sqrt{v_2},\;w_3 = \sqrt{v_3}.
\]
Solving the system (\ref{eq:Two_p}) for $\tau_2$ yields
\begin{equation}
\begin{array}{l}
\displaystyle{{\tau_2}^{(j)}(n) = \frac{1}{w_n}\left[\arctan\left({\frac {a_{{1}}{w_n}^{5}+ \left( b_{{1}}a_{{2}}-c_{{1}}-a_{{1}}b_{{2}}
 \right) {w_n}^{3}+ \left( c_{{1}}b_{{2}}-b_{{1}}c_{{2}} \right) w_n}{
 \left( b_{{1}}-a_{{1}}a_{{2}} \right) {w_n}^{4}+ \left( 
c_{{1}}a_{{2}}+a_{{1}}c_{{2}} -b_{{1}}b_{{2}}\right) {w_n}^{2}-c_{{1}}c_{{2}}}}\right)
+ (j-1)\pi\right],}\\\\
\\
n = 1,2,3;\;\; j\in \mathbb{N}.
\end{array}
\end{equation}
This allows us to define the following:
\begin{equation}\label{tauzero}
{\tau_2}^* = {\tau_2}^{(j_0)}(n_0)=\min_{1\le n\le 3,\;j\ge 1}\{{\tau_2}^{(j)} (n)\}, w_0 = w_{n_0}.
\end{equation}

In order to establish whether the steady state $E_2$ actually undergoes a Hopf bifurcation at $\tau_2={\tau_2}^*$, we compute the sign of $d[\operatorname{Re}\mu({\tau_2}^*)]/d \tau_2$. Differentiating both sides of equation (\ref{eq:Tran}) with respect to $\tau_2$ yields
\begin{equation*}
\left(\frac{d\mu}{d\tau_2}\right)^{-1} = {\frac { \left( 3\,{\mu}^{2}+2\,a_{{2}}\mu+b_{{2}} \right) {{e}^{\mu\,\tau_{{2}}}}+2\,a_{{1}}\mu+b_{{1}}}{\mu\, \left( a_{{1}}{\mu}^{2}
+b_{{1}}\mu+c_{{1}} \right) }}-{\frac {\tau_{{2}}}{\mu}}.
\end{equation*}
Introducing the notation $V = {w_0}^{2} \left[ \left(c_{{1}} -{w_0}^{2}a_{{1}} \right) ^{2}+{w_0}^{2}{b_{{1}}}^{2} \right]$, it follows that $V>0$ for all $w_0>0$, and
\begin{equation*}
\left(\frac{d\operatorname{Re}\mu({\tau_2}^*)}{d \tau_2}\right)^{-1} =\frac{w_0}{V} \left[\underbrace{A\cos(w_0\tau_2)+wB\sin(w_0\tau_2)}_{:=\Gamma}-{b_{{1}}}^{2}w_0+2\,a_{{1}}w_0 \left( c_{{1}} -{w_0}^{2}a_{{1}}\right)\right], 
\end{equation*}
where
\begin{eqnarray}
\nonumber A &=& \left( 3\,{w_0}^{2}-b_{{2}} \right) b_{{1}}w_0-2\,w_0a_{{2}} \left( {w_0}^{2
}a_{{1}}-c_{{1}} \right),
 \\
\nonumber B &=& 2\,{w_0}^{2}a_{{2}}b_{{1}}+ \left( 3\,{w_0}^{2}-b_{{2}} \right)  \left( 
{w_0}^{2}a_{{1}}-c_{{1}} \right),
\end{eqnarray}
and
\[
\Gamma = 3\,{w_0}^{5}+ \left( 2\,{a_{{2}}}^{2}-4\,b_{{2}} \right) {w_0}^{3}+
\left( {b_{{2}}}^{2} -2\,a_{{2}}c_{{2}}\right) w_0
\]
Consequently, for $v_0 = {w_0}^2$ we have
\begin{equation}
\begin{array}{l}
\displaystyle{\left(\frac{d\operatorname{Re}\mu({\tau_2}^*)}{d \tau_2}\right)^{-1} =\frac{1}{V}\Big[3{w_0}^{6}+ 2\left({a_{{2}}}^{2}-{a_{{1}}}^{2}-2b_{{2}}
\right){w_0}^{4}+ \left( 2a_{{1}}c_{{1}}-2a_{{2}}c_{{2}}+{b_{{2}}}^{2}-{b_{{1}}}
^{2} \right) {w_0}^2\Big]}\\\\
\displaystyle{=\frac{1}{V}\left[3{w_0}^6+2p{w_0}^4+q{w_0}^2\right]=\frac{1}{V}\left[3v_0^3+2pv_0^2+qv_0\right]=\frac{v_0}{V}h'(v_0),}
\end{array}
\end{equation}
where $h(v)$ is defined in (\ref{eq:cubic}). Since $v_0 = {w_0}^2>0$, this implies
\begin{equation*}
\mathrm{sign}\left(\frac{d\operatorname{Re}\mu({\tau_2}^*)}{d \tau_2}\right) =  \mathrm{sign}\left(\frac{d\operatorname{Re}\mu({\tau_2}^*)}{d \tau_2}\right)^{-1} =\mathrm{sign} \left[v_0h'(v_0)\right]=\mathrm{sign} \left[h'(v_0)\right].
\end{equation*}
These calculations can now be summarised as the following theorem.\\

\begin{theorem}\label{EndemicSS} Let the coefficients of the characteristic equation (\ref{eq:Tran}) satisfy $a_1+a_2>0, c_1+c_2>0$ and  $(a_1+a_2)(b_1+b_2)>c_1+c_2$. Additionally, let  $w_0,{\tau_2}^*$ be defined as in (\ref{tauzero}) with $v_0={w_0}^2$, and let $h'(v_0)> 0$. Then, the following holds.\\
\noindent (i) If $r\ge 0$ and $p^2<3q$, or $p^2>3q$ and $v_1<0$, or $p^2>3q$ and $v_1<0$ and $h(v_1)>0$, the endemic steady state $E_2$ of the system (\ref{eq:system_two}) is linearly asymptotically stable for all $\tau_2\ge 0$.\\\\
\noindent (ii) If $r<0$, or if $r\ge 0$ and $h(v_1 )<0$, the endemic steady state $E_2$ of system (\ref{eq:system_two}) is linearly asymptotically stable when $\tau_2 \in [0,{\tau_2}^*)$, unstable for $\tau_2>\tau_2^*$ and undergoes Hopf bifurcation at $\tau_2=\tau_2^*$. 
\end{theorem}

\subsection{Fast-spreading PTGS signal}

In the case when the PTGS signal is spreading very quickly, the time delay $\tau_2$ associated with the spread of this signal is negligibly small compared to other timescales in the system. In this case, setting $\tau_2 = 0$ in the characteristic equation (\ref{mother}) reduces it to
\begin{equation}\label{cubic_t2_zero}
\mu^3+\left[a_1(\tau_1)e^{-\mu\tau_1} +a_2(\tau_1)\right]\mu^2 +\left[b_1(\tau_1)e^{-\mu\tau_1} +b_2(\tau_1)\right] \mu +c_1(\tau_1)e^{-\mu\tau_1} +c_2(\tau_1) = 0,
\end{equation}
where
\[
\begin{array}{l}
a_1(\tau_1) =  - K(\tau_1),\hspace{0.5cm}
a_2(\tau_1) =  \left( \delta+\lambda+\delta\,\phi \right) {\it I^*}+2\,\epsilon\,S^*+\epsilon,\\
b_1(\tau_1) =  - \left[ \epsilon+\delta\,(\phi+1)\,{\it I^*} \right] 
K\left( {\it \tau_1} \right),\\
b_2(\tau_1) =  \delta\,\phi\left( \lambda\,+{\delta}\right) {{\it I^*}}^{2
}+ \left[S^*\left[  \left( \lambda+2\,\phi\,\epsilon \right) \delta+{
\lambda}^{2} \right] + \left( \epsilon+\phi\,\epsilon \right) \delta+
\lambda\,\epsilon \right] {\it I^*}+2\,{\epsilon}^{2}S^*,\\
c_1(\tau_1) =  -K\left( {\it \tau_1} \right) \, \left( {\it I^*}\,
\delta\,\phi+\phi\,\epsilon+\lambda\,S^* \right)\delta\,{\it I^*},\\
c_2(\tau_1) =  \delta \phi \epsilon\left( {\delta}+\lambda\,
 \right) {{\it I^*}}^{2}+ \left[ \left( 2\,{\epsilon}^{2}\phi+
\lambda\,\epsilon \right) \delta+{\lambda}^{2}\epsilon \right] S^*{\it 
I^*}.
\end{array}
\]

Looking for solutions of equation (\ref{cubic_t2_zero}) in the form $\mu = iw$ $(w>0)$, and separating the real and imaginary parts gives
\begin{equation}\label{sys:t2_zero}
\begin{array}{l}
b_1(\tau_1)w\sin(w\tau_1)-[a_1(\tau_1)w^2-c_1]\cos(w\tau_1) = [a_2(\tau_1)w^2-c_2(\tau_1)],\\\\
b_1(\tau_1)w\cos(w\tau_1)+[a_1(\tau_1)w^2-c_1]\cos(w\tau_1) = [w^3-b_2(\tau_1)w].
\end{array}
\end{equation}
With the help of auxiliary functions
\[
\begin{array}{l}
g_1(\tau_1) = b_1(\tau_1)w, \hspace{0.5cm} g_2(\tau_1) = a_1(\tau_1)w^2-c_1(\tau_1),\\\\
L_1(\tau_1) = a_2(\tau_1)w^2-c_2(\tau_1), \hspace{0.5cm} L_2(\tau_1) = w^3-b_2(\tau_1)w,
\end{array}
\]
the system of equations (\ref{sys:t2_zero}) can be re-written as follows
\begin{equation}\label{eq:primal_B}
\begin{array}{l}
g_1(\tau_1)\sin(w\tau_1)-g_2(\tau_1)\cos(w\tau_1) = L_1(\tau_1),\\\\
g_1(\tau_1)\cos(w\tau_1)+g_2(\tau_1)\sin(w\tau_1) = L_2(\tau_1).
\end{array}
\end{equation}
Solving this system yields
\begin{gather}
\begin{aligned}
& {\tau_1}^{(j)}(n) = \frac{1}{w_n}\left[{\tan}^{-1}U(\tau_1)+ (j-1)\pi\right],\qquad n = 1,2,3;\;\; j\in \mathbb{N}, \\
& U(\tau_1) = \frac{g_1(\tau_1)L_1(\tau_1)+L_2(\tau_1)g_2(\tau_1)}{L_2(\tau_1)g_1(\tau_1)-L_1(\tau_1)g_2(\tau_1)},
\end{aligned}
\end{gather}
 though, unlike the case of instantaneous maturity, $w_n$ is now itself the function of $\tau_1$, and hence, it does not prove possible to find the closed form expression for the critical time delay $\tau_1^*$.\\

\noindent{\bf Remark.} In the case when maturation delay and the PTGS propagation delay coincide, i.e. $\tau_1=\tau_2=\tau$, the characteristic equation (\ref{mother}) once again becomes an equation with a single time delay. However, similar to the case we have just considered, the critical value of the time delay can only be found implicitly, as the coefficients of the characteristic equation  themselves depend on the time delay. In the case where both $\tau_1>0$ and $\tau_2>0$, application of a methodology discussed in Gu {\it et al.} \cite{Gu2005} and Blyuss {\it et al.} \cite{Blyuss2008}, would provide a parameterisation of critical time delays but such parameterisation would also be implicit.

\section{Numerical stability analysis and simulations}

In order to gain a better insight into how different parameters affect biological feasibility and stability of different steady states, as well as to understand the dynamics inside stability regions, especially when $\tau_{1,2}>0$, we use a Matlab package traceDDE \cite{Breda2006} to numerically compute eigenvalues of the characteristic equation (\ref{mother}). Since PTGS is known to be a complex multi-component process, obtaining accurate values for parameters to be used in the model is very problematic, especially since there is a significant variation in reported values for many of the parameters, and some cannot currently be measured \cite{Melnyk,Liang, Himber15}. In light of this, we accompany theoretical analysis from the previous sections by an extensive numerical bifurcation analysis of the model to illustrate different types of behaviour that can be exhibited when the system parameters are varied. This provides qualitative insights into possible dynamics, which can be further improved once more advanced measurement techniques are developed, and the precise mechanisms of PTGS are elucidated.

\begin{figure}[h]
\centerline{\includegraphics[scale=0.6]{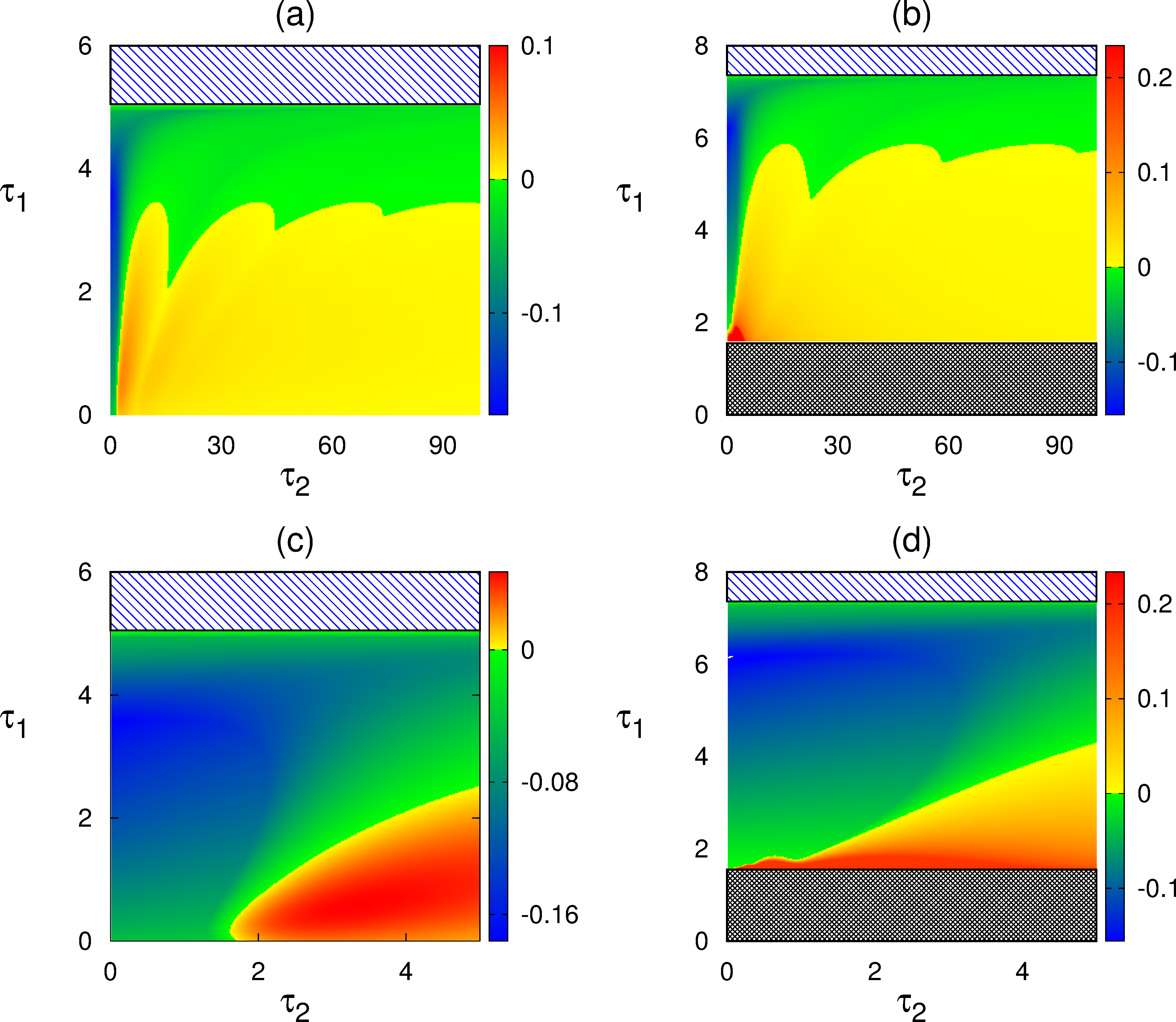}}
\caption{Stability of the endemic and disease-free steady states with parameter values from Table \ref{tab:List of parameters}.  (a) and (c) $k = 1$. (b) and (d) $k = 2$. Diagonal blue indicates the region where the disease-free steady state is asymptotically stable, and the endemic steady state is not feasible. The black  grid shows the region where the endemic steady state is not feasible, and none of the steady states is stable. Colour code denotes max[Re$(\mu)$] for the endemic steady state when it is feasible.}
\label{fig:2x2}
\end{figure}

Figure \ref{fig:2x2} shows the regions of stability of the disease-free steady state, as well as feasibility and stability of the endemic steady state. For parameter values specified in Table \ref{tab:List of parameters} and $k=1$, it follows from Theorems \ref{Theorem:Disease_free} and \ref{EndemicSS} that the endemic steady state is only feasible for $\tau_1\in \left[0,5.05\right)$, whereas for $\tau_1 \ge 5.05$, the endemic steady state disappears, and the disease-free steady state becomes asymptotically stable, as shown in Figs.~\ref{fig:2x2} (a) and (c). When the growth rate $k$ is increased, a qualitatively similar picture is observed, however, there is some minimum value of $\tau_1$, below which the endemic steady state is not biologically feasible. Figs.~\ref{fig:2x2} (b) and (d) illustrate that in this case, the endemic steady state is only feasible for $\tau_1\in \left[1.54,7.36\right)$, and for $\tau_1 \ge 7.36$ the disease-free steady state is asymptotically stable. This figure suggests that by adequately increasing the time delay $\tau_2$ after which susceptible cells  acquire immunity, the endemic steady state can generally become unstable, whereas there are regions in which the solution of the system alternates between the stable endemic steady state and  solutions of a periodic or possibly  chaotic nature. From a biological perspective, this is an interesting and a rather surprising result since intuitively one would expect that increasing the time delay associated with the spread of PTGS signal (i.e. time necessary to acquire immunity) would promote stabilization of the endemic steady state.
\begin{figure}
\centerline{\includegraphics[scale=0.6]{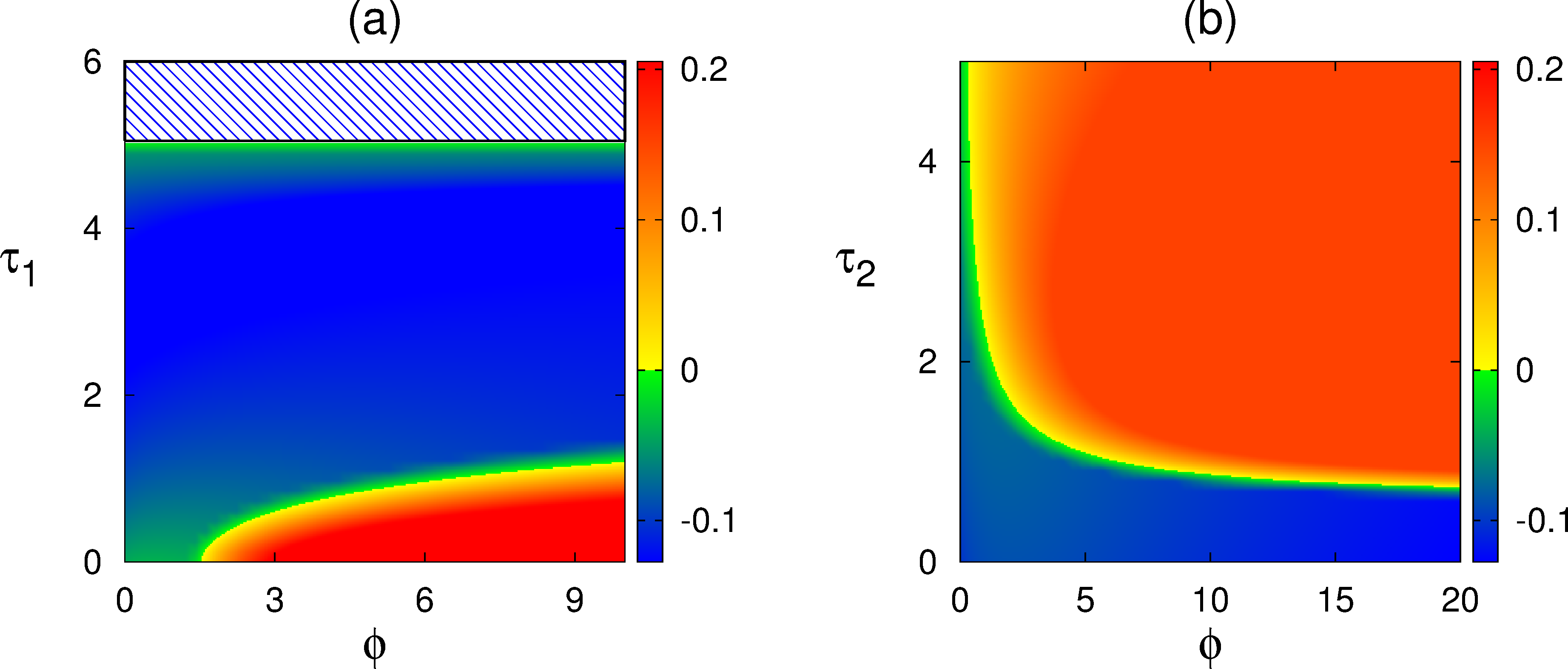}}
\caption{Stability of the endemic and disease-free steady states with parameter values from Table \ref{tab:List of parameters}. Diagonal blue indicates the region where the disease-free steady state is asymptotically stable, and the endemic steady state is not feasible. Colour code denotes max[Re$(\mu)$] for the endemic steady state when it is feasible.}
\label{fig:D}
\end{figure}

As a next step, we investigate how the relative values of the time delays and the amplification factor $\phi$, affect stability of the steady states. Figure~\ref{fig:D}(a) shows that the endemic steady state, when feasible, is asymptotically stable for sufficiently high values of the maturation delay $\tau_1$, but can lose stability once $\tau_1$ becomes lower than some critical value that is itself increasing with $\phi$. This implies that both the higher amplification factor and the faster maturation of the new plant tissue are prone to make the endemic steady state, characterised by some permanent level of infection, unstable. Figure~\ref{fig:D}(b) demonstrates the above-mentioned counter-intuitive result, which suggests that the endemic steady state is stable only for sufficiently fast-spreading PTGS signal, i.e. sufficiently small $\tau_2$. Due to the functional form of the term representing recovery of infected cells associated with the spreading PTGS signal, it is natural to expect that the critical time delay $\tau_2$ would be inversely proportional to $\phi$, and this is indeed what is observed in Fig.~\ref{fig:D}(b). It is noteworthy that in the parameter region where the endemic steady state is unstable, the disease-free steady state is also unstable. This highlights one of our earlier conclusions, namely that the amplification of recovery by the propagation of the warning signal, which in this case is transmitted from infected cells to other infected cells, has a limited impact on the outcome of the infection. Moreover, it is not by itself sufficient to achieve complete annihilation of the virus from its plant host.

\begin{figure}
\centerline{\includegraphics[scale=0.6]{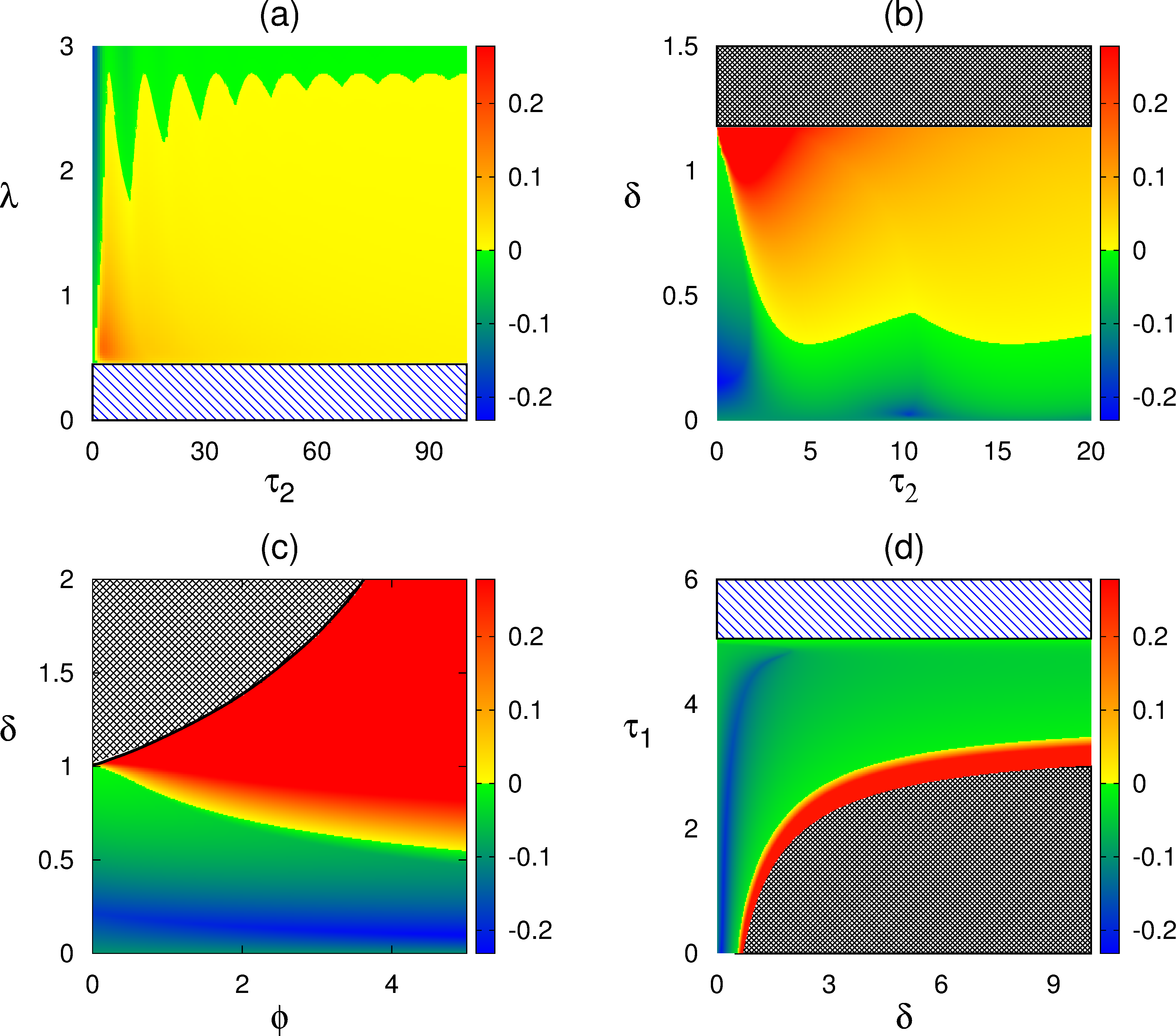}}
\caption{Stability of the endemic and disease-free steady states with parameter values from Table \ref{tab:List of parameters}. Diagonal blue indicates the region where the disease-free steady state is asymptotically stable, and the endemic steady state is not feasible. The black  grid shows the region where the endemic steady state is not feasible, and none of the steady states is stable. Colour code denotes max[Re$(\mu)$] for the endemic steady state when it is feasible.}
\label{fig:F}
\end{figure}
Figure~\ref{fig:F} shows that if the infection rate $\lambda$ is sufficiently small, or if the maturation of the growing tissue is sufficiently slow (i.e. $\tau_1$ is large), the disease-free steady state is asymptotically stable. On the other hand, if the infection rate is high, the endemic steady state is asymptotically stable, and the PTGS propagation delay $\tau_2$ becomes irrelevant to the long-term behaviour of the system. One can observe that for a sufficiently small warning rate $\delta$, the endemic steady state can be asymptotically stable for any value of $\tau_2$, whereas if $\delta$ is large enough, neither endemic, nor disease-free steady states are stable. The same happens in the case when the new plant tissue is maturing fast, i.e. $\tau_1$ is sufficiently small.

\begin{figure}
\centerline{\includegraphics[scale=0.6]{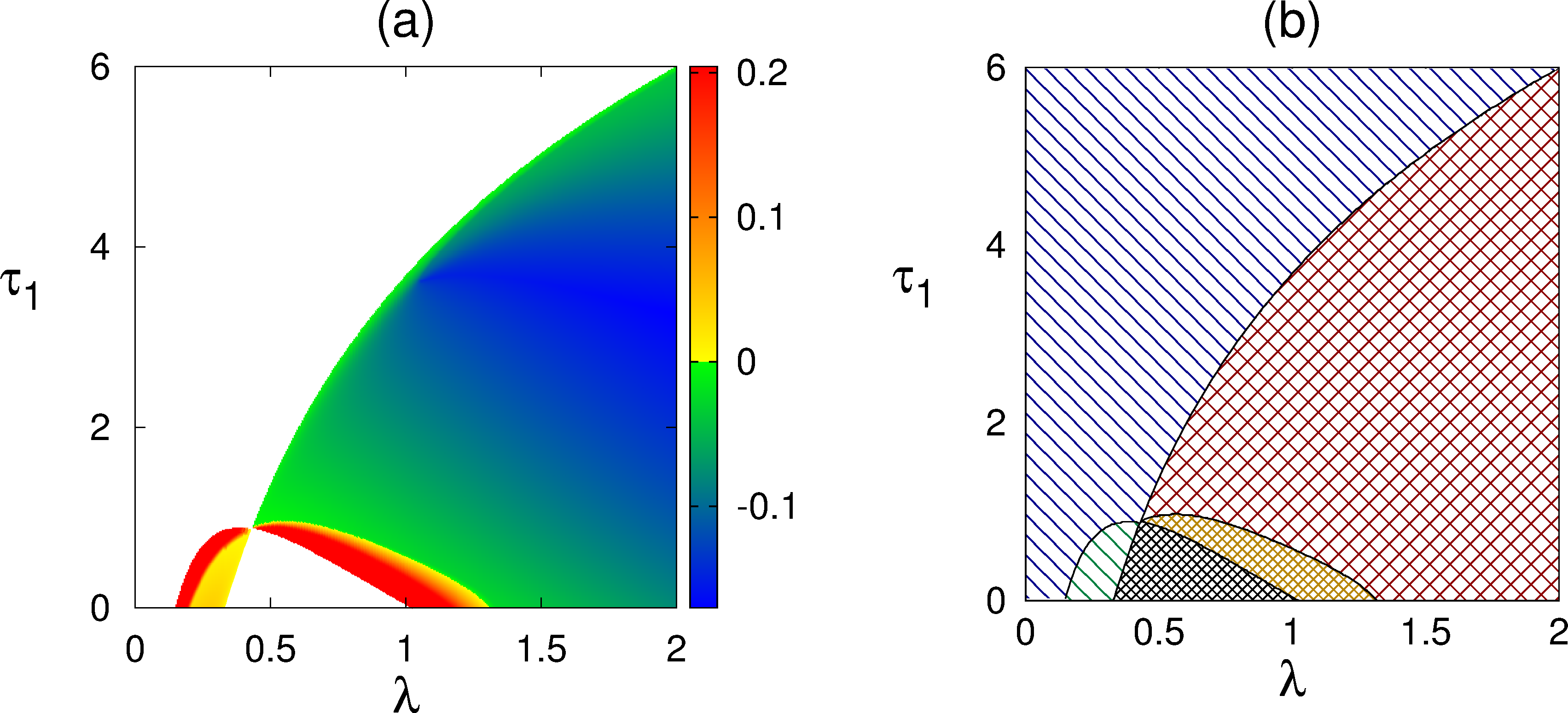}}
\caption{(a) Colour code denotes max[Re$(\mu)$] for the endemic steady state when it is feasible. (b) Stability regions of all steady states with parameter values from Table \ref{tab:List of parameters}. The black  grid shows the region where the endemic steady state is not feasible, and none of the steady states is stable. The area covered with  diagonal lines signifies the region where the disease-free steady state is asymptotically stable; in the region with green diagonal lines all steady states are feasible, whereas for  blue lines  the endemic steady state is not feasible. The red grid represents the area for which the endemic steady state is asymptotically stable. The brown grid shows  the region where both the endemic and disease free steady state are feasible but none are stable.}
\label{fig:C}
\end{figure}

In Figure \ref{fig:C} we have used the results from Theorems \ref{Theorem:Disease_free} and \ref{EndemicSS} to identify regions in which the system transitions from a  stable disease-free to the endemic steady state. When all other parameters remain fixed, this figure suggests that there is a minimum value of $\lambda$ for which the endemic steady state $E_2$ is asymptotically stable provided that the time delay $\tau_1$ is small enough. However, for any value of $\lambda$ below that threshold, either the system reverts back to the stable disease-free steady state, or the time-delay $\tau_1$ has to be within a specific range for the endemic steady state to be feasible and asymptotically stable. Our results up to this point suggest that $\tau_1$ is perhaps the most important bifurcation parameter in the model. From a biological perspective, this can be explained by interpreting the time delay $\tau_1$ as a temporary immunity inherent to the nature of proliferating and undifferentiated cells responsible for new growth. Equivalently, these results imply that whether or not the disease can successfully take over the plant depends on how fast the virus can gain access to  the newly formed parts of the plant. If the infection rate is not sufficiently high, the infected parts of the plant will eventually die out before the newer generation of cells becomes vulnerable to infection.

\begin{figure}
\begin{subfigure}{.5\textwidth}
  \centerline{\includegraphics[scale=0.63]{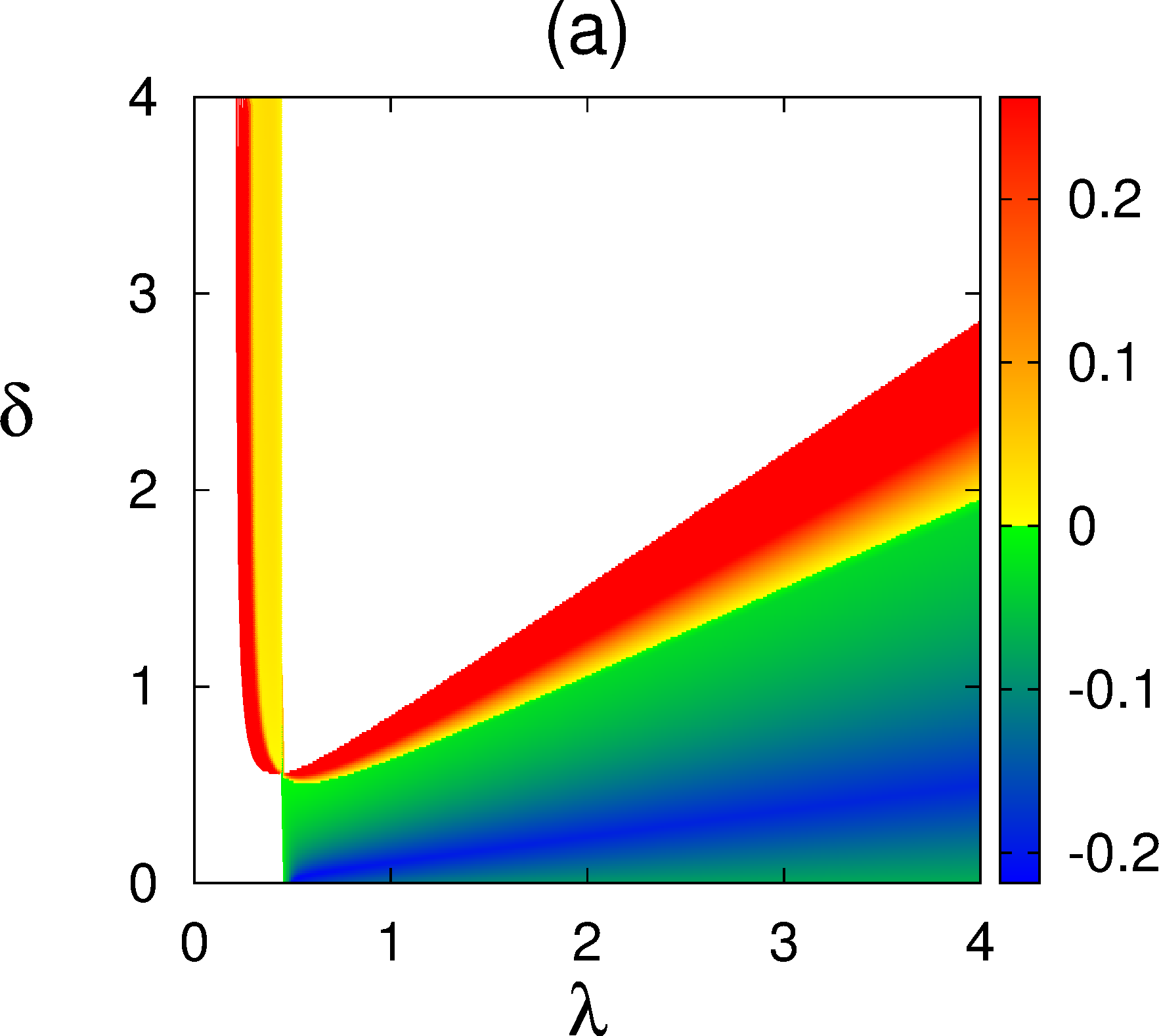}}
\end{subfigure}%
\hspace{0.3cm}
\begin{subfigure}{.5\textwidth}
  \centerline{\includegraphics[scale=0.67]{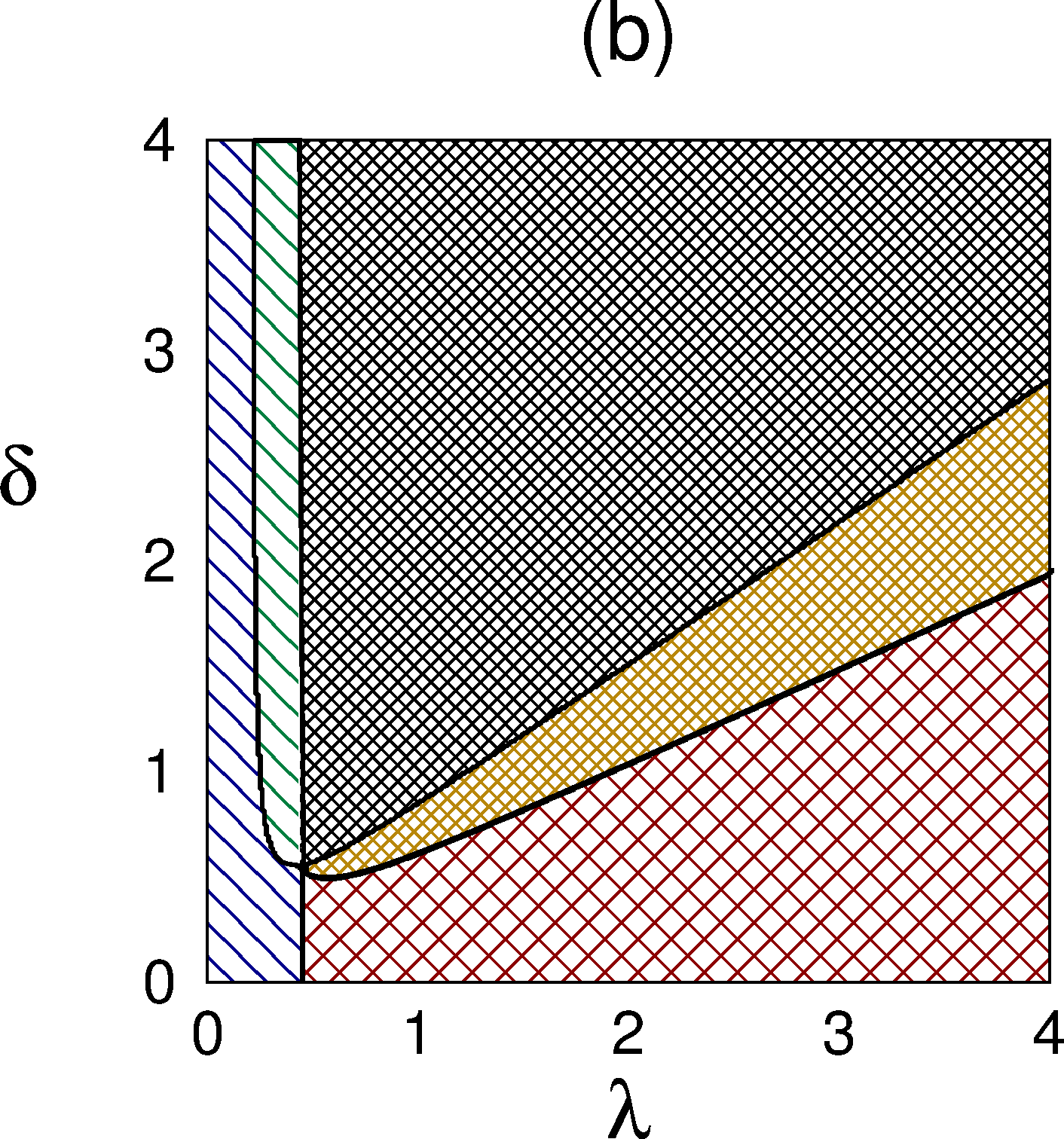}}
\end{subfigure}\\\\
\vspace*{0.7cm}
\begin{subfigure}{\textwidth}
  \centerline{\includegraphics[scale=0.65]{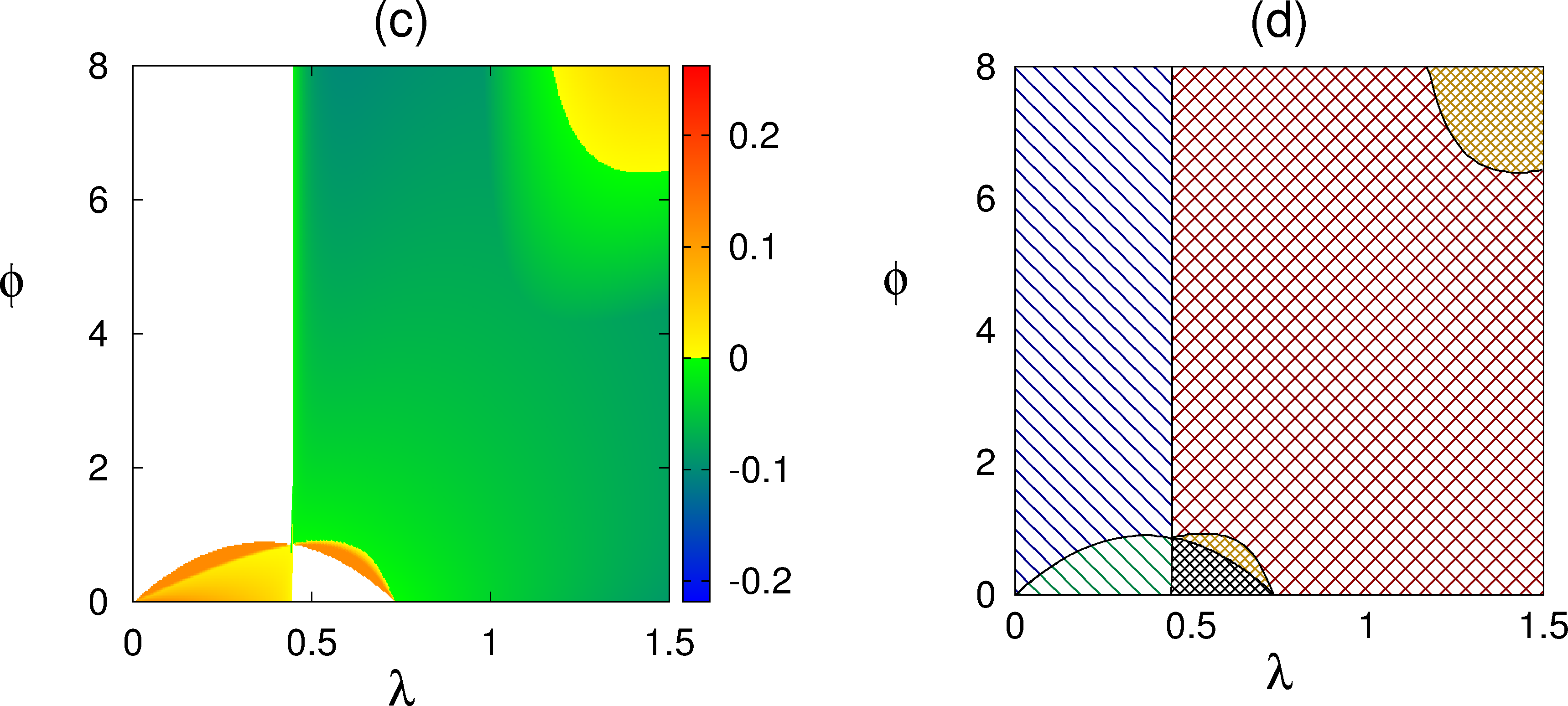}}
\end{subfigure}
\caption{(a), (c) Colour code denotes max[Re$(\mu)$] for the endemic steady state when it is feasible. (b), (d) Stability regions of all steady states with parameter values from Table \ref{tab:List of parameters}. The black  grid shows the region where the endemic steady state is not feasible, and none of the steady states is stable. The area covered with  diagonal lines signifies the region where the disease-free steady state is asymptotically stable; in the region with green diagonal lines all steady states are feasible, whereas for  blue lines  the endemic steady state is not feasible. The red grid represents the area for which the endemic steady state is asymptotically stable. The brown grid shows  the region where both the endemic and disease free steady state are feasible but none are stable.}
\label{fig:A}
\end{figure}

Figure~\ref{fig:A} illustrates the regions of feasibility and stability of the disease-free and endemic steady states when the time delays are fixed, and other parameters are allowed to vary. Naturally, the disease-free steady state is stable for lower values of the disease transmission rate $\lambda$, while for higher $\lambda$ there is a propensity for the endemic steady state to be stable. Higher speed of propagation of the PTGS signal $\delta$ and higher amplification factor $\phi$ lead to a de-stabilisation of the endemic steady state. It is worth noting the behaviour shown in Figs.~\ref{fig:A}(c) and (d), where for sufficiently high amplification rate, increase in the disease transmission rate $\lambda$ also destabilises the endemic steady state.

To illustrate different types of dynamical behaviour that can be exhibited by the model (\ref{eq:system_two}), we solve this system numerically with parameter values given in Table \ref{tab:List of parameters} and different values of the time delays $\tau_1$ and $\tau_2$. Figures~\ref{fig:num}(a), (c) and (d) demonstrate partial immune response that is not sufficient to eradicate the virus in the host. This is the type of behaviour one might expect from susceptible plants with a weak response against a viral disease, and it results in a chronic condition. Another possibility for a chronic infection is represented by periodic solutions shown in Fig.~\ref{fig:num}(b) and (e), where the severity of infection varies over time, with periods of high viral production being interspersed with periods of quiescence. From a biological perspective, these scenarios could be interpreted as situation where the evolutionary race between viral pathogen and the host immune system has not yet concluded, and as a result neither the plants immune system nor the virus' ability to suppress immune responses can prevail. Figure~\ref{fig:num}(f) demonstrates a type of immune response consistent with a recovery phenotype, where initially the disease appears to overwhelm the plant by infecting a dominating or a rather significant part of its body. However, as the warning signal propagates to surrounding cells, newly grown tissue and uninfected cells are able to acquire immunity and thus prevent the spread of the disease. This localizes the infection and eventually leads to the eradication of the invading virus, and consequently the system approaches a disease-free steady state. Similar type of behaviour is observed in the system with a very strong immunity response that would be consistent with highly resistant plants; in this case the infection is almost immediately localized due to the high efficacy of the propagating warning signal and the antiviral activity in the cells that are already infected.

\begin{figure}
\hspace{-0.5cm}
\includegraphics[width=14.5cm]{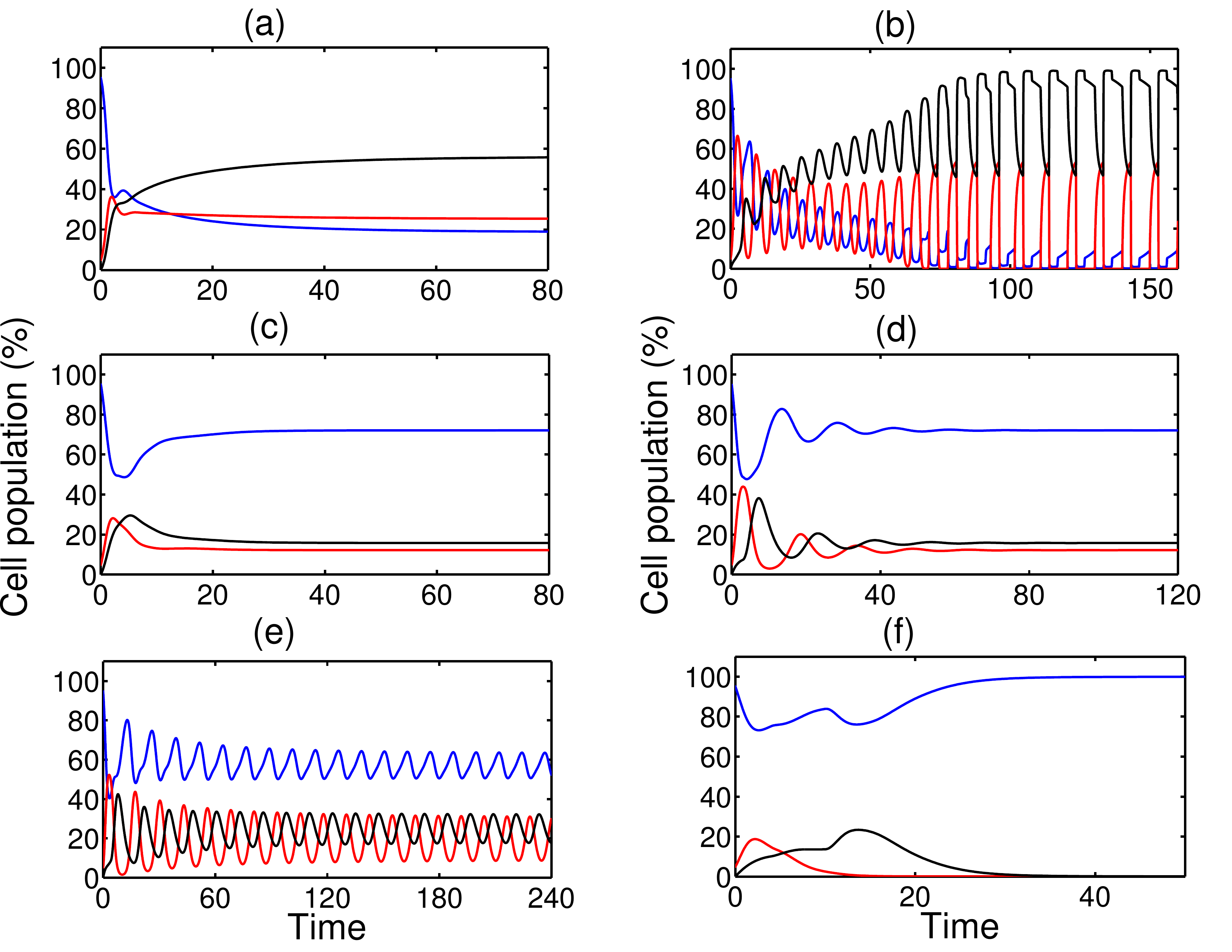}
\caption{Numerical solution of the system (\ref{eq:system_two}) with parameter values  from Table \ref{tab:List of parameters}. (a) $\tau_1=\tau_2=0$. (b) $\tau_1=0$, $\tau_2=3$. (c) $\tau_1=3$, $\tau_2=0$. (d) $\tau_1=\tau_2=3$. (e) $\tau_1=2$, $\tau_2=4$.
(f) $\tau_1=4$, $\tau_2=10$, $\sigma=1$, $\phi=0.1$. Colours represent scaled populations of susceptible $S$ (blue), infected $I$ (red) and warned $W$ cells (black).}
\label{fig:num}
\end{figure}

\section{Discussion}

In this paper we have developed and analysed a new mathematical model of the plant immune response to a viral infection, with particular emphasis on the role of RNA interference. To achieve better biological realism, this model explicitly includes two different time delays, one to represent maturation period of undifferentiated cells which effectively acts as a form of inherent immunity against infection, and another to account for the time required for the PTGS signal to reach other parts of the plant resulting in either recovery or warning of susceptible cells.

Stability analysis of the model has demonstrated the role played by system parameters in the dynamics. In the present model, it is impossible for all plant cells to die due to the constant emergence of new susceptible cells. Stability of the disease-free steady state appears to depend only on the maturation period but not on the speed of propagation or the strength of the PTGS signal, suggesting that a faster PTGS signal can at most help the plant to recover faster, but by itself it is not sufficient for a recovery. Endemic steady state, where the plant supports some constant level of infection, is only biologically feasible when the growth rate of the new tissue is higher than some minimum value. An interesting and counter-intuitive result is that slower PTGS signal (i.e. larger value of $\tau_2$) can actually lead to a destabilisation of this steady state, resulting in sustained periodic oscillations. Another possibility for the endemic steady state to lose stability is when the amplification factor $\phi$ increases, or the new uninfected tissue is produced faster, i.e. for a lower maturation time delay $\tau_1$.

Numerical simulations have shown that the model can support resistant- and recovery-type behaviours, whereby the plant immune system is able to mount sufficient response to eradicate the infection. Both of these situations are characterised by a strong localised immune response, but if additionally the warning signal is sufficiently strong, the plant exhibits the resistant phenotype, where the spread of infection is almost fully prevented, and the amount of the virus is diminished significantly faster than in the recovery case. On the other hand, if both the localised immune response and the propagating signal are sufficiently weak, the plant will be very susceptible to infection, however, the infection cannot result in the death of the host in our model. Periodic solutions of the model signify specific cases where the plant immune system cannot mount a sufficient response to eradicate the virus, and at the same time the virus also cannot adequately suppress the immune response of the plant. As a result, the plant undergoes periods of time in which the symptoms of the disease are manifested more prominently, with other periods where the infection is at a very low level.

Simulations suggest that the propagating component of the PTGS has a very limited impact on the long-term recovery of the plant. At the same time, the duration of maturation period of undifferentiated cells can play a very important role in controlling the spread of the infection, as it represents how fast the newly developed part of the plant becomes accessible to the virus. An interesting and practically important question is whether the model can be further improved by including some more realistic distribution of maturation periods in the way it has been done when modelling different distributions of temporary immunity \cite{BK10,YB14}, latency and incubation \cite{BT95,MCC}, or infectious periods \cite{RH07,ZLZ08}.

\section*{Acknowledgements} The authors would like to thank two anonymous reviewers for helpful comments and suggestions that have helped to improve the presentation in this paper.


\begin{thebibliography}{50}

\bibitem{Tilman2011} Tilman, D., Balzer, C., Hill, J., Befort, B.L., 2011. Global food demand and the sustainable intensification of agriculture. Proc. Natl. Acad. Sci. USA 108, 20260--20264.

\bibitem{Strange2005} Strange, R.N., Scott, P.R., 2005. Plant disease: a threat to global food security. Ann. Rev. Phytopath. 43, 83--116.

\bibitem{Jeger2004} Jeger, M.J., Holt, J., Van Den Bosch, F., Madden, L.V., 2004. Epidemiology of insect-transmitted plant viruses: modelling disease dynamics and control interventions. Physiol. Ent. 29, 291--304.

\bibitem{Purcell2005} Purcell, A.H., Almeida, R.P.P., 2005. Insects as vectors of disease agents, in Dekker, M. (Ed.), Encyclopedia of plant and crop science. Taylor and Francis, London.

\bibitem{Chan1994} Chan, M.S., Jeger, M.J., 1994. An analytical model of plant virus disease dynamics with roguing and replanting. J. Appl. Ecol. 31, 413--427.

\bibitem{VandenBosch1996} van den Bosch, F., de Roos, A.M., 1996. The dynamics of infectious diseases in orchards with roguing and replanting as control strategy. J. Math. Biol. 35, 129--157.

\bibitem{Zhang2012a} Zhang, T., Meng, X., Song, Y.,  Li, Z., 2012. Dynamical analysis of delayed plant disease models with continuous or impulsive cultural control strategies. Abstr. Appl. Anal. 2012, 1--25.

\bibitem{Zhang2000} Zhang, X.-S., Holt, J., Colvin, J., 2000. A general model of plant-virus disease infection incorporating vector aggregation. Plant Path. 49, 435--444.

\bibitem{Zhang2000a} Zhang, X.-S., Holt, J., Colvin, J., 2000. Mathematical models of host plant infection by helper-dependent virus complexes: why are helper viruses always avirulent? Phytopath. 90, 85--93.

\bibitem{JD06} Jones, J.D.G., Dangl, J.L., 2006. The plant immune system. Nature 444, 323--329.

\bibitem{gerg06} Gergerich, R.C., Dolya, V.V., 2006. Introduction to plant viruses, the invisible foe. The Plant Health Instructor.

\bibitem{frit87} Fritig, B., Kauffmann, S., Dumas, B., Geoffroy, P., Kopp, M., Legrand, M., 1987. Mechanism of the hypersensitivity reaction of plants, in Evered, D., Harnett, S., Plant resistance to viruses. Wiley, New York.

\bibitem{Waterhouse1999} Waterhouse, P.M., Smith, N.,A. Wang, M.B., 1999. Virus resistance and gene silencing: killing the messenger. Trends Plant Sci. 4, 452--457.

\bibitem{Escobar2000} Escobar, M.A., Dandekar, A.M., 2003. Post-transcriptional gene silencing in plants, in Barciszewski, J.,  Erdmann, V.A., Noncoding RNAs: molecular biology and molecular medicine. Kluwer Academic, Dordrecht.

\bibitem{VBF01} Vaucheret, H., B\'eclin, C., Fagard, M., 2001. Post-transcriptional gene silencing in plants. J. Cell. Sci. 114, 3083--3091.

\bibitem{Waterhouse2001} Waterhouse, P.M., Wang, M.B., Lough, T., 2001. Gene silencing as an adaptive defence against viruses. Nature 411, 834--842.

\bibitem{AlessandraTenorioCosta2013} Costa, A.T., Bravo, J.P., Makiyama, R.K., Vasconcellos Nunes, A., Maia, I.G., 2013. Viral counter defines X antiviral immunity in plants: mechanisms for survival, in Romanowski, V., Current issues in molecular virology - viral genetics and biotechnological applications, InTech.

\bibitem{Hammond2000} Hammond, S.M., Bernstein, E., Beach, D., Hannon, G.J., 2000. An RNA-directed nuclease mediates post-transcriptional gene silencing in {\it Drosophila} cells. Nature 404, 293--296.

\bibitem{Bernstein2001} Bernstein, E., Caudy, A.A., Hammond, S.M., Hannon, G.J., 2001. Role for a bidentate ribonuclease in the initiation step of RNA interference. Nature 409, 363--366.

\bibitem{Bergstrom2003} Bergstrom, C.T.,  McKittrick, E., Antia, R., 2003. Mathematical models of RNA silencing: unidirectional amplification limits accidental self-directed reactions. Proc. Natl. Acad. Sci. USA 100, 11511--11516.

\bibitem{Raab2004} Raab, R.M., Stephanopoulos, G., 2004. Dynamics of gene silencing by RNA interference. Biotech. Bioengrg 88, 121--132.

\bibitem{Groenenboom2005} Groenenboom, M.A.C, Mar\'ee, A.F.M., Hogeweg, P., 2005. The RNA silencing pathway: the bits and pieces that matter. PLoS Comp. Biol. 1, 155--165.

\bibitem{Cuccato2011} Cuccato, G., Polynikis, A., Siciliano, V., Graziano, M., di Bernardo, M., di Bernardo, D., 2011. Modeling RNA interference in mammalian cells. BMC Syst. Biol. 5, 19.

\bibitem{Raja2008} Raja, P., Sanville, B.C., Buchmann, R.C., Bisaro, D.M., 2008. Viral genome
methylation as an epigenetic defense against geminiviruses. J. Virol. 82, 8997--9007.

\bibitem{Mitsuhara2002} Mitsuhara, I., Shirasawa-Seo, N., Iwai, T., Nakamura, S., Hokura, R., Ohashi, Y., 2002. Release from post-transcriptional gene silencing by cell proliferation in transgenic tobacco plants: possible mechanism for noninheritance of the silencing. Genetics 160, 343--352.

\bibitem{MiassarM.Al-Taleb2011} Al-Taleb, M.M., Hassawi, D.S., Abu-Romman, S.M., 2011. Production of virus free potato plants using meristem culture from cultivars grown under Jordanian environment. American-Eurasian J. Agric. Environ. Sci. 11, 467--472.

\bibitem{Taskn2013} Ta\c{s}kin, H., Baktemur, G., Kurul, M., B\"uy\"ukalaca, S., 2013. Use of tissue culture
techniques for producing virus-free plant in garlic and their identification through real-time PCR. Sci. World J. 2013, 781282.


\bibitem{Molnar2011} Molnar, A., Melnyk, C., Baulcombe, D.C., 2011. Silencing signals in plants: a long journey for small RNAs. Gen. Biol. 12, 215.

\bibitem{Wassenegger2000} Wassenegger, M., 2000. RNA-directed DNA methylation. Plant Mol. Biol. 43, 203--220.

\bibitem{Zhang2012} Zhang, C., Ruvkun, G., 2012. New insights into siRNA amplification and RNAi. RNA Biol. 9, 1045--1049.

\bibitem{Sonoda2000} Sonoda, S., Nishiguchi, M., 2000. Delayed activation of post-transcriptional gene silencing and de novo transgene methylation in plants with the coat protein gene of sweet potato feathery mottle potyvirus. Plant Sci. 156, 137--144.


\bibitem{Roth2004} Roth, B., 2004. Plant viral suppressors of RNA silencing. Virus Res. 102, 97--108.


\bibitem{Canizares2008} Ca\~{n}izares, M.C., Navas-Castillo, J., Moriones, E., Multiple suppressors of RNA silencing encoded by both genomic RNAs of the crinivirus, Tomato chlorosis virus. Virology 369, 168--174.

\bibitem{Llave2000} Llave, C., Kasschau, K.D., Carrington, J.C., 2000. Virus-encoded suppressor of posttranscriptional gene silencing targets a maintenance step in the silencing pathway. Proc. Natl. Acad. Sci. USA 97, 13401--13406.


\bibitem{Burgyan2011} Burgy\'{a}n, J., Havelda, Z., 2011. Viral suppressors of RNA silencing. Trends in Plant Sci. 16, 265--272.

\bibitem{Zvereva2012} Zvereva, A.S., Pooggin, M.M., 2012. Silencing and innate immunity in plant defense against viral and non-viral pathogens. Viruses 4, 2578--2597.

\bibitem{TF2002}Tenllado, F., D\'iaz-Ru\'izâ J.R., 2001. Double-stranded RNA-mediated interference with plant virus infection. J. Virol. 75, 12288--12297.

\bibitem{Camp61} Campbell, A., 1961. Conditions for the existence of bacteriophage. Evolution 15, 153--165. 

\bibitem{BK01} Beretta, E. Kuang, Y., 2001. Modeling and analysis of a marine bacteriophage infection with latency period. Nonl. Anal. RWA 2, 35--74.

\bibitem{PN01} Perelson, A.S., Nelson, P.W., 1999. Mathematical analysis of HIV-1: dynamics in vivo. SIAM Rev. 41, 3--44.

\bibitem{GK05} Gourley, S.A., Kuang, Y., 2005. A delay reaction-diffusion model of the spread of bacteriophage infection. SIAM J. Appl. Math. 65, 550--566.

\bibitem{LS03} De, Lenheer, P., Smith, H.L., 2003. Virus dynamics: a global dynamics. SIAM J. Appl. Math. 63, 1313--1327.

\bibitem{SJC13} Schultz, J.C., Appel, H.M., Ferrieri, A.P.,  Arnold, T.M. (2013). Flexible resource allocation during plant defense responses. Front. Plant Sci. 4, 324.

\bibitem{Berg07} Berger, S., Sinha, A. K., Roitsch T., 2007. Plant physiology meets phytopathology: plant primary metabolism and plant-pathogen interactions. J. Exp. Bot. 58, 4019--4026. 

\bibitem{Rojas2014} Rojas, C.M., Senthil-Kumar, M., Tzin, V.,  Mysore, K.S., 2014. Regulation of primary plant metabolism during plant-pathogen interactions and its contribution to plant defense. Front. Plant Sci. 5, 17.

\bibitem{Ehn1997} Ehness R., Ecker M., Godt D.E., Roitsch T., 1997. Glucose and stress independently regulate source and sink metabolism and defense mechanisms via signal transduction pathways involving protein phosphorylation. Plant Cell 9, 1825--1841.

\bibitem{Smith95} Smith, H.L., 1995. Monotone dynamical systems: an introduction to the theory of competitive and cooperative systems. American Mathematical Society, Providence.

\bibitem{Smith11} Smith, H., 2011. An introduction to delay differential equations with applications to life sciences. Springer, New York.

\bibitem{Ruan2001} Ruan, S., Wei, J., 2001. On the zeros of a third degree exponential polynomial with applications to a delayed model for the control of testosterone secretion. Math. Med. Biol. 18, 41--52.

\bibitem{Gu2005} Gu, K., Niculescu, S.-I., Chen, J., 2005. On stability crossing curves for general systems with two delays. J. Math. Anal. Appl. 311. 231--253.

\bibitem{Blyuss2008} Blyuss, K.B., Kyrychko, Y.N., H\"ovel, P., Sch\"oll, E., 2008. Control of unstable steady states in neutral time-delayed systems. Eur. Phys. J. B 65, 571--576.

\bibitem{Breda2006} Breda, D., Maset, S., Vermiglio, R., 2006. Pseudospectral approximation of eigenvalues of derivative operators with non-local boundary conditions. Appl. Num. Math. 56, 318--331.

\bibitem{Melnyk} Melnyk, C.W., Molnar, A. Baulcombe, D.C., 2011. Intercellular and systemic movement of RNA silencing signals. EMBO J. 30, 3553--3563.

\bibitem{Liang} Liang, D., White, R.G., Waterhouse, P.M., 2012. Gene silencing in {\it Arabidopsis} spreads from the root to the shoot, through a gating barrier, by template-dependent, nonvascular, cell-to-cell movement. Plant Physiol. 159, 984--1000.

\bibitem{Himber15} Himber, C., Dunoyer, P., 2015. The tracking of intercellular small RNA movement. Meth. Mol. Biol. 1217, 275-281.

\bibitem{BK10} Blyuss, K.B., Kyrychko, Y.N., 2010. Stability and bifurcations in an epidemic model with varying immunity period. Bull. Math. Biol. 72, 490--505.

\bibitem{YB14} Yuan, Y., B\'elair, J., 2014. Threshold dynamics in an SEIRS model with latency and temporary immunity. J. Math. Biol. 69, 875--904.

\bibitem{BT95} Beretta, E., Takeuchi, Y., 1995. Global stability of an SIR epidemic model with time delays. J. Math. Biol. 33, 250--260.

\bibitem{MCC} McCluskey, C.C., 2010. Global stability of an SIR epidemic model with delay and general nonlinear incidence. Math. Biosci. Eng. 7, 837-850.

\bibitem{RH07} Roberts, M.G., Heesterbeek, M.G., 2007. Model-consistent estimation of the basic reproduction number from the incidence of an emerging infection. J. Math. Biol. 55, 803-816.

\bibitem{ZLZ08} Zhang, F., Li, Z., Zhang, F., 2008. Global stability of an SIR epidemic model with constant infectious period. Appl. Math. Math. 199, 285--291.

\end{thebibliography}
\end{document}